\documentclass[12pt]{amsart}
\usepackage{amsmath, amssymb, amsbsy, amsfonts, amsthm, latexsym, amsopn, amstext, amsxtra, euscript, amscd, color, mathrsfs}
\usepackage{amsmath,amscd}
\usepackage[normalem]{ulem}
\usepackage{soul}

\makeatletter
\@namedef{subjclassname@2020}{%
  \textup{2020} Mathematics Subject Classification}
\makeatother
\PassOptionsToPackage{hyphens}{url}
\usepackage[hidelinks]{hyperref}

\usepackage{tabularx}
\usepackage[colorinlistoftodos,prependcaption,textsize=tiny]{todonotes}

\usepackage{amscd}
\usepackage{color,enumerate}
\usepackage{array, multirow, multicol}
\usepackage{longtable}
\usepackage{bm}  
\newcommand{\RNum}[1]{\lowercase\expandafter{\romannumeral #1\relax}}

\newtheorem{thm}{Theorem}[section]
\newtheorem{lem}[thm]{Lemma}

\newtheorem{exmp}[thm]{Example}

\newtheorem{rmk}[thm]{Remark}

\newtheorem{thm-con}[thm]{Theorem-Conjecture}
\numberwithin{equation}{section}

\theoremstyle{definition}
\newtheorem{defn}[thm]{Definition}

\newcommand{\F}{\mathbb F}
\newcommand{\bP}{\mathbb P}

\def\Tr{{\rm Tr}}

\begin{document}
\title[Permutation polynomials with  a few terms over finite fields]{Permutation polynomials with  a few terms \\ over finite fields}

\author[K. Garg]{Kirpa Garg}
\address{Department of Mathematics, Indian Institute of Technology Jammu, Jammu 181221, India}
\email{2020RMA1030@iitjammu.ac.in}
\author[S. U. Hasan]{Sartaj Ul Hasan}
\address{Department of Mathematics, Indian Institute of Technology Jammu, Jammu 181221, India}
\email{sartaj.hasan@iitjammu.ac.in}
\author[C. Li]{Chunlei Li}
\address{Department of Informatics, University of Bergen, PB 7803, N-5020, Bergen, Norway}
\email{chunlei.li@uib.no}
\author[H. Kumar]{Hridesh Kumar}
\address{Department of Mathematics, Indian Institute of Technology Jammu, Jammu 181221, India}
\email{2021RMA2022@iitjammu.ac.in}
\author[M. Pal]{Mohit Pal}
\address{Department of Informatics, University of Bergen, PB 7803, 5020, Bergen, Norway}
\email{mohit.pal@uib.no}
\thanks{The work of S. U. Hasan was supported by the Science and Engineering Research Board (SERB), Government of India (Grant No. CRG/2022/005418), and the Indo-Norwegian Cooperation Programme 2024 (Project No. INCP2-2024/10213). He also acknowledges the support of SPIRE project from Univ. of Bergen, Norway. H. Kumar acknowledges support from the Prime Minister’s Research Fellowship (PMRF), Government of India, under PMRF ID 3002900 at IIT Jammu. C. Li’s work was supported by the Research Council of Norway (Grant No. 311646) and the Indo-Norwegian Cooperation Programme 2024 (Project No. INCP2-2024/10213). M. Pal acknowledges support from the Research Council of Norway under Grant No. 314395.}

\begin{abstract}

This paper considers permutation polynomials over the finite field $\F_{q^2}$ by utilizing low-degree permutation rational functions over $\F_q$, where $q$ is a prime power.
As a result, we obtain two classes of permutation binomials, six classes of permutation quadrinomials and six classes of permutation pentanomials over $\F_{q^2}$. Additionally, we show that the obtained binomials, quadrinomials and pentanomials are quasi-multiplicative inequivalent to 
the known ones in the literature.
\end{abstract}
\maketitle

\section{Introduction}
Let $\F_q$ be the finite field with $q=p^m$ elements, where $p$ is a prime number and $m$ is a positive integer. We denote by $\F_q^*$ the multiplicative cyclic group of non-zero elements of the finite field $\F_q$ and by $\F_q[X]$ the ring of polynomials in one variable $X$ with coefficients in $\F_q$. Let $f$ be a function from the finite field $\F_q$ to itself. It is elementary and well-known that $f$ can be uniquely represented by a polynomial in $\F_q[X]$ of degree strictly less than $q$. A polynomial $f(X)\in \F_q[X]$ is called a permutation polynomial (PP) if the induced mapping $X \mapsto f(X)$ is a bijection of $\F_q$. Permutation polynomials over finite fields have been an interesting area of research for many years due to their wide range of applications in coding theory~\cite{ Ding_C_13, Chapuy_C_07}, cryptography~\cite{Bertoni_Daemen,  Dwokin_SHA3, Lidl_Cr_84, Schwenk_Cr_98}, combinatorial design theory~\cite{Ding_Co_06}, and several other branches of mathematics and engineering. For a survey of recent developments in permutation polynomials over finite fields, the reader may refer to~\cite{Hou_PP_15,WangIndex}.

Permutation polynomials with a few terms are of particular interest due to their simple algebraic structures,  which make them suitable for implementation in various applications. A lot of research has been done in recent years on the construction of permutation polynomials with a few terms. The most simple polynomials are the monomials $X^n$, which permutes $\F_q$ if and only if $\gcd(n,q-1)=1$. This fact makes the classification of permutation monomials easy.
On the other hand, the classification of permutation polynomials having a few terms, such as binomials, trinomials, quadrinomials, and pentanomials, is non-trivial and has not yet been completely resolved.
In order to study permutation binomials, Hou argued in \cite{Hou_bi_tri} that it is enough to consider binomials of the form $f(X)=X^r(X^{\frac{q-1}{d}}+a)$, where $ r \geq 1, d \mid (q-1)$ and $a\in \F_q^{*}$. For $r=1$ and a fixed $d$, Carlitz and Wells ~\cite{CW} proved that there always exist elements $a\in \F_q^*$ such that $f(X)$ permutes $\F_q$ for sufficiently large $q$. A lot of research~\cite{ JLTZ, WL,MZ,OM} has been done on the characterization of $r’$s and $a’$s, for which $f(X)$ is a permutation binomial. For  binomials of the form $f(X) = X^n +aX^m$ Turnwald in~\cite{TG}, and further, Masuda and Zieve in~\cite{MZ} have given certain conditions on the positive integers $m, n$, and $a \in \F_q^{*}$ for which $f(X)$ are not permutations of $\F_q$. Further, the permutation binomials having the form $X(X^{r(q-1)}+ a)$ over $\F_{q^2}$ for $r \in \{2,3,5,7\}$ are investigated in ~\cite{HXX,HL,LS}. For explicit constructions of permutation binomials and trinomials so far, the interested reader may refer to ~\cite{BZ,BS,SG,HV, LFW,LQ,Wang_2007}.

Unlike permutation binomials and trinomials, relatively little is known about permutation quadrinomials over finite fields. Permutation quadrinomials with Niho-like exponents have attracted a lot of attention because of
their good cryptographic properties ~\cite{KLi_quad_2021, NLi_quad_2021}. Gupta~\cite{Gupta20} constructed several families of permutation quadrinomials of $\F_{q^2}$ with coefficients derived from the subfield $\F_q$ of characteristic 3 or 5. The author continued this work in~\cite{Gupta22} and gave more general classes of permutation quadrinomials where the coefficients were not restricted to $\F_q$ only. 
Recently, {\"O}zbudak and Tem{\"u}r~\cite{FB_quad_2023} provided a complete characterization for a class of permutation quadrinomials over $\F_{q^2}$, where $q$ is an odd prime power and the coefficients are coming from the finite field $\F_q$. For more results on permutation quadrinomials, the reader is referred to~\cite{ Chen, KLi_quad_2020, NLi_quad_2021,Gupta_resi_2023, ZTu_quad_19,  ZTu_quad_18_2, ZTu_quad_18_1, Zheng_quad_22} and references therein.

The study of permutation pentanomials has recently attracted growing interest and attention. For more than a decade, there was only one class of permutation pentanomials, which was given by Dobbertin~\cite{Dobbertin_niho} to prove Niho’s conjecture. Later, Xu, Cao and Ping~\cite{Xu_penta_18} constructed several classes of permutation pentanomials with trivial coefficients derived from fractional polynomials that permute the unit circle of $\F_{q^2}$ with order $q+1$, where $q=2^m$. 
Moreover, the authors in~\cite{LQW} proposed six new classes of permutation pentanomials over $\F_{2^{2m}}$.
Liu, Chen, Liu and Zou~\cite{QLiu_penta_2023} (for $p=2,3$) investigated a few more classes of permutation pentanomials with general coefficients over $\F_{p^{2m}}$. Recently, Rai and Gupta~\cite{Rai} characterized three new classes of permutation pentanomials over $\F_{q^2}$ where $q$ is an even prime power and coefficients are in $\F_q$.
For more details on permutation pentanomials, we refer the readers to the papers~\cite{Deng, KX, Shen, Zhang_penta_23} and references therein.

In this paper, we shall present several new classes of permutation polynomials (binomials, quadrinomials and pentanomials) of the form $X^rh(X^{q-1})$ for $r\geq 1$ by studying some low-degree polynomials $h(X)$, which are derived from low-degree rational functions that permute the unit circle $\mu_{q+1}$ of $\F_{q^2}$ with order $q+1$. 
Our method is based on the fractional polynomials which permute the projective line $\bP^1(\F_q)= \F_q \cup \{ \infty \}$ and will be used to construct low-degree bijections on $\mu_{q+1}$.
By applying the criterion \cite{Park_2001, Wang_2007, Zieve2008}  for $X^rh(X^{q-1})$ to permute $\F_{q^2}$, we derive several classes of permutation polynomials over $\F_{q^2}$.

The remainder of this paper will be organized as follows. In Section~\ref{S2}, we recall some definitions and lemmas that will be used in later sections. In Section ~\ref{S3}, we classify permutation quadrinomials from degree-three rational functions. Further, we give six classes of permutation pentanomials and two classes of permutation binomials derived from degree-four rational functions in Section ~\ref{S4}. In Section~\ref{S5}, we 
investigate the quasi-multiplicative equivalence among the obtained permutation polynomials and the already known classes of permutation polynomials. We conclude our work in Section~\ref{S6}.

\section{Preliminaries}\label{S2}
In this section, we review some definitions and provide several lemmas that will be used in the subsequent sections. 

A function of the form $f(X)=\frac{P(X)}{Q(X)} $ defines a mapping from $\bP^1(\F_q)=\F_q\cup\{\infty\}$ to itself, where $P(X), Q(X) \in \F_q[X]$ and $\gcd(P,Q)=1$, is called \textit{rational function}. We denote the degree of $f$ by $\deg(f)$, i.e.,$\deg(f):=\max(\deg(P), \deg(Q))$. The set of all such rational functions is denote by $\F_q(X)$. We call $f(X)$ a permutation rational (PR) function if the induced map $X \mapsto f(X)$ is a bijection from $\bP^1(\F_q)$ to itself.
\begin{defn}\cite{Ding2020}
Two rational functions $f$ and $g$ are equivalent if there exist degree-one rational functions $\phi$ and $\psi$ such that $g= \phi \circ f \circ \psi$.
\end{defn}
 It is easy to see that for equivalent functions $f,g\in \F_q(X)$, $f$ permutes $\bP^1(\F_q)$ if and only if $g$ permutes  $\bP^1(\F_q)$.
We now recall some lemmas that give us the complete classification of permutation rational functions of degree at most four over $\F_q$.

\begin{lem}\cite{Ding2020}\label{L20}
Every degree-one $f(X)\in \F_q(X)$ permutes $\bP^1(\F_q)$. A degree-two $f(X)\in \F_q(X)$ permutes $\bP^1(\F_q)$ if and only if $q$ is even and $f(X)$ is equivalent to $X^2$.
\end{lem}
The following theorem from~\cite{Ding2020} (see also \cite{Ferraguti2020,Hou_carlitz}) gives a complete classification of permutation rational functions of degree-three over $\F_q$.
 \begin{lem}\cite[Theorem 1.3]{Ding2020}
\label{T20}
  A degree-three $f(X) \in \F_q(X)$ permutes $\bP^1(\F_q)$ if and only if it is equivalent to the following
  \begin{enumerate}
   \item $X^3$~\mbox{where}~$q\equiv2 \pmod 3$
   \item $\zeta^{-1} \circ X^3 \circ \zeta$~\mbox{where}~$q\equiv1 \pmod 3$ and for some $\delta \in \F_{q^2}\setminus \F_q$ we have $\zeta(X)=(X-\delta^q)/(X-\delta)$ and $ \zeta^{-1}(X)=(\delta X-\delta^q)/(X-1)$
   \item $X^3 -\alpha X$ where $3 \mid q$ and either $\alpha =0$ or $\alpha$ is a non-square in $\F_q$.
  \end{enumerate}
 \end{lem}
 Motivated by the results of~\cite{Ferraguti2020}, Hou~\cite{HouD4}, Ding and Zieve~\cite{Ding2020} classified degree-four permutation rational functions of $\bP^1(\F_q)$. The following lemma from~\cite{Ding2020} (see also~\cite{HouD4}) gives a complete classification of degree-four permutation rational functions of $\bP^1(\F_q)$.
 \begin{lem}\cite[Theorem 1.4]{Ding2020}
 \label{T21}
A degree-four $f(X)\in \F_q(X)$ permutes $\bP^1(\F_q)$ if and only if one of the following holds.
\begin{enumerate}
   \item $q$ is even and $f(X)$ is equivalent to $X^4+\alpha X^2+\beta X$ for some $\alpha, \beta \in \F_q$ such that $X^3+\alpha X+ \beta$ has no roots in $\F_q^*$
   \item $q\leq 8$ and $f(X)$ is equivalent to a rational function in Table $1$~\cite[Page 19]{Ding2020}
   \item $q$ is odd and $f(X)$ is equivalent to $$\frac{X^4-2\alpha X^2-8\beta X+\alpha^2}{X^3+\alpha X+\beta}$$ for some $\alpha, \beta \in \F_q$ such that $X^3+\alpha X+ \beta$ is irreducible in $\F_q[X].$
  \end{enumerate}
\end{lem}
Zieve in~\cite{Zieve2013} introduced a generic construction of permutation polynomials in which the induced function on $\mu_k$ was represented by R\' edei functions, namely, rational functions over a field which are conjugate to $X^n$ over an extension field. The main ingredient of this construction are the following lemmas that provide necessary and sufficient condition for a degree-one rational function to be a bijection from $\mu_{q+1}$ to $\bP^1(\F_q)$.
\begin{lem}\label{L21}
 Let $\rho(X) \in \F_{q^2}(X)$ be a degree-one rational function. Then $\rho(X)$ induces a bijection from $
\mu_{q+1}$ to $\bP^1(\F_q)$ if and only if $\displaystyle \rho(X)= \frac{\delta X -\beta \delta^q}{X-\beta}$ with $\beta \in \mu_{q+1}$ and $\delta \in \F_{q^2}\backslash \F_q$.
\end{lem}
\begin{lem}\label{L22}
 Let $\nu(X) \in \F_{q^2}(X)$ be a degree-one rational function. Then $\nu(X)$ induces a bijection from $\bP^1(\F_q)$ to $\mu_{q+1}$ if and only if $\displaystyle \nu(X)= \tilde{\beta}\frac{X -\tilde{\delta^q}}{X-\tilde{\delta}}$ with $\tilde{\beta} \in \mu_{q+1}$ and $\tilde{\delta} \in \F_{q^2}\backslash \F_q$.
\end{lem}
We shall now describe our strategy to construct new permutation polynomials over $\F_{q^2}$ from the permutation rational functions over $\F_q$. Firstly, by using bijective maps $\rho$ and $\nu$ as discussed in Lemma~\ref{L21} and Lemma~\ref{L22}, respectively, and also by making use of permutation rational function $f$, we shall explicitly determine the expression of $\nu \circ f \circ \rho$ which is a bijection on $\mu_{q+1}$ as shown in the following diagram.

$$\begin{CD}
\bP^1(\F_q) @>f>> \bP^1(\F_q)\\
 @A\rho AA   @VV \nu V\\
 \mu_{q+1}  @>\nu \circ f \circ
 \rho>>  \mu_{q+1}
\end{CD}$$

Secondly, we shall use such bijective maps $\nu \circ f \circ \rho$  to construct new permutation polynomials over $\F_{q^2}$ by using the following lemma that demonstrates a general method of producing permutation polynomials. This lemma was studied in various forms; see, for example, Wan and Lidl ~\cite{WL}, Park and Lee ~\cite{Park_2001}, Akbary and Wang ~\cite{AW}, Wang ~\cite{Wang_2007}, and Zieve ~\cite{Zieve2008}. 
In order to utilize Lemma  ~\ref{L23}, our construction strategy requires $\nu \circ f \circ \rho$  to be of the form $X^rh(X)^s$, which will be investigated for the employed low-degree functions.
\begin{lem}\label{L23}
 Let $h(X) \in \F_q[X]$ and $r,s$ are positive integers such that $s \mid (q-1)$. Then the function $f(X)=X^rh(X^s)$ permutes $\F_q$ if and only if both conditions
 \begin{enumerate}
  \item $\gcd(r,s)=1$ and
  \item $X^rh(X)^s$ permutes $\mu_{\frac{q-1}{s}}$,
 \end{enumerate}
are satisfied, where $\mu_{\frac{q-1}{s}}$ is the set of $\frac{q-1}{s}$-th roots of unity in the algebraic closure $\overline \F_q$ of $\F_q$.
\end{lem}

It is worthy to mention that the above stated lemma is a special case of multiplicative version of AGW criterion~\cite{AGW}.

Recently, Wu, Yuan, Ding and Ma~\cite{wu2017} introduced the concept of quasi-multiplicative equivalence (QM-equivalence) between the permutation polynomials. This equivalence preserves both the bijection property and number of terms of a polynomial over $\F_q$, and 
has recently gained increasing attention in the literature. We also discuss the QM-equivalence of permutation polynomials proposed in this paper to the known ones in the literature.
\begin{defn}\label{QM-equivalence}
Two permutation polynomials $f(X)$ and $g(X)$ in $\F_q [X]$ are called quasi-multiplicative equivalent if there exists an integer $1 \leq  d <  q-1$ with
$\gcd(d, q-1) = 1$ and $f(X) = ug(vX^d)$, where $u, v \in \F_{q}^{*}$.
\end{defn}

\begin{rmk}\label{R1} Notice that two equivalent permutation rational functions $f$ and $g$ from $\mathbb{P}^1(\F_q)$ to itself shall produce  same classes of PPs over $\F_{q^2}$. Thus, it is sufficient to consider a representative from each permutation class of rational functions.
\end{rmk}

With the complete classification of permutation rational functions over $\F_q$ in Lemmas ~\ref{L20}-\ref{L22}, we will propose several classes of permutations with at most five terms in the subsequent sections.
As binomials obtained from degree-one and degree-two are linearized and well studied in the literature, we will mainly focus on degree-three and degree-four rational functions.

\section{Permutation polynomials from degree-three rational functions} \label{S3}
In this section, we will consider degree-three permutation rational functions permuting $\bP^1(\F_q)$, which will give rise to six classes of permutation quadrinomials over $\F_{q^2}$.
In view of Lemma~\ref{T20}, we split our analysis in the following  three cases.
\subsection*{Case 1:} \label{C41} Assume $f(X)\in \F_q(X)$ permutes $\bP^1(\F_q)$, where $q\equiv2\pmod3$. In this case, we shall construct permutation quadrinomials over $\F_{q^2}$ arising from degree-three PRs of $\bP^1(\F_q)$.
From Lemma~\ref{T20}, notice that, if $q\equiv 2 \pmod 3$ then any degree-three PR of $\bP^1(\F_q)$ will be equivalent to $X^3$. Let $\rho : \mu_{q+1} \rightarrow \bP^1(\F_q)$ and ${\nu} : \bP^1(\F_q) \rightarrow \mu_{q+1}$ are bijective degree-one rational functions. Then from Lemma~\ref{L21} and Lemma ~\ref{L22}, we know that
 \[
  \rho(X) := \frac{\delta X -\beta \delta^q}{X-\beta},
 ~\mbox{and}~
  {\nu}(X) := {\tilde \beta} \left(\frac{X -{\tilde \delta}^q}{X- {\tilde \delta}}\right),
 \]
 for some $\beta, {\tilde \beta} \in \mu_{q+1}$ and $\delta, {\tilde \delta} \in \F_{q^2} \backslash \F_q$. It is straightforward to see that the rational function ${\nu} \circ X^3 \circ \rho : \mu_{q+1} \rightarrow \mu_{q+1}$ is a bijection of $\mu_{q+1}$ and is given by
 \begin{equation*}
  \begin{split}
   & {\nu} \circ \left(\frac{\delta X -\beta \delta^q}{X-\beta} \right)^3 \\
   &= {\tilde \beta} \left(\frac{(\delta X -\beta \delta^q)^3 -{\tilde \delta}^q (X -\beta)^3}{(\delta X -\beta \delta^q)^3- {\tilde \delta}(X -\beta)^3}\right)\\
   &= {\tilde \beta} \left(\frac{(\delta^3-{\tilde \delta}^q) X^3  -3\beta (\delta^{q+2}-{\tilde \delta}^q)X^2+3\beta^2(\delta^{2q+1}-{\tilde \delta}^q)X -\beta^3 (\delta^{3q}-{\tilde \delta}^q)}{(\delta^3-{\tilde \delta}) X^3  -3\beta (\delta^{q+2}-{\tilde \delta})X^2 +3\beta^2(\delta^{2q+1} -{\tilde \delta})X -\beta^3 (\delta^{3q}-{\tilde \delta})}\right)\\
   &= \frac{N_3X^3+N_2X^2+N_1X+N_0}{D_3X^3+D_2X^2+D_1X+D_0},
  \end{split}
 \end{equation*}
 where
 \begin{equation}\label{C05}
  \begin{cases}
   N_0 &= -{\tilde \beta}\beta^3 (\delta^{3q}-{\tilde \delta}^q),\\
   N_1 &= 3{\tilde \beta}\beta^2(\delta^{2q+1}-{\tilde \delta}^q),\\
   N_2 &= -3{\tilde \beta}\beta (\delta^{q+2}-{\tilde \delta}^q),\\
   N_3 &= {\tilde \beta}(\delta^3-{\tilde \delta}^q),
  \end{cases}
  ~~~\mbox{and}~~~
  \begin{cases}
   D_0 &= -\beta^3 (\delta^{3q}-{\tilde \delta}),\\
   D_1 &= 3\beta^2(\delta^{2q+1}-{\tilde \delta}),\\
   D_2 &= -3\beta (\delta^{q+2}-{\tilde \delta}),\\
   D_3 &= (\delta^3-{\tilde \delta}).
  \end{cases}
 \end{equation}

 \begin{rmk}\label{R05}
  Since $ {\nu} \circ X^3 \circ \rho: \mu_{q+1} \rightarrow \mu_{q+1}$ is a permutation of $\mu_{q+1}$, the equation $$h(X):= D_3X^3+D_2X^2+D_1X+D_0 = 0$$
  has no solution $X \in \mu_{q+1}$.
 \end{rmk}

 We shall now use Remark~\ref{R05} to prove the following lemma which will be used to construct PPs over $\F_{q^2}$.

\begin{lem} \label{L05}
 Let
\begin{equation*}
 \begin{split}
  h(X)&=D_3X^3 +D_2X^2+D_1X+D_0,\\
  h_1(X) &= D_0X^q+D_3X^2 +D_2X+D_1,\\
  h_2(X) &=D_0X^{2q}+D_1X^q+D_3X+D_2,\\
  h_3(X) &=D_0X^{3q}+D_1X^{2q}+D_2X^q+D_3,
 \end{split}
\end{equation*}
where $D_0,D_1, D_2, D_3$ are as defined in the Equation~\eqref{C05}. Then the following statements are equivalent.
\begin{enumerate}[(a)]
 \item $h(X)=0$ has no solution in $\mu_{q+1}$,  \item $h_1(X)=0$ has no solution in $\mu_{q+1}$,  \item $h_2(X)=0$ has no solution in $\mu_{q+1}$,  \item $h_3(X)=0$ has no solution in $\mu_{q+1}$.
\end{enumerate}
\end{lem}
\begin{proof} We shall include here the proof only for the first equivalence. Let $h(X)=0$ has no solution in $\mu_{q+1}$. On the contrary, assume that $h_1(X)=0$ has a solution $\alpha \in \mu_{q+1}$, i.e., $h_1(\alpha)=0$. Then $h_1(\alpha)=\alpha^qh(\alpha)=0$ implies that $h(\alpha)=0$, a contradiction. Thus, $h(X)=0$ has no solution in $\mu_{q+1}$ implies that  $h_1(X)=0$ has no solution in $\mu_{q+1}$. Conversely, let $h_1(X)=0$ has no solution in $\mu_{q+1}$. On the contrary, assume that $h(X)=0$ has a solution $\alpha \in \mu_{q+1}$, i.e., $h(\alpha)=0$. Then $h(\alpha)=\alpha h_1(\alpha)=0$ implies that $h_1(\alpha)=0$, a contradiction. Thus, $h_1(X)=0$ has no solution in $\mu_{q+1}$  implies that  $h(X)=0$ has no solution in $\mu_{q+1}$. Hence statements (a) and (b) are equivalent. Similar arguments can be used to prove the remaining equivalence.
\end{proof}

We now can use the rational function ${\nu} \circ X^3 \circ \rho$ to construct permutation quadrinomials over $\F_{q^2}$.

\begin{thm}\label{T05}
Let $q\equiv 2 \pmod 3$ and $h(X)=D_3X^3+D_2X^2 +D_1X+D_0$, where $D_0, D_1, D_2,$ $D_3$ are as given in Equation~\eqref{C05},
and $\beta, {\tilde \beta}$ be elements in $\mu_{q+1}$ satisfying $1+{\tilde \beta}\beta^3=0$. Moreover, $\delta$ and $\tilde \delta$ in the expressions of $D_i$'s, $0 \leq i \leq 3$ are such that $\delta, \tilde \delta \in \F_{q^2}\setminus\F_{q}$ and $\tilde \delta \not \in \{ \delta^3, \delta^{q+2}, \delta^{2q+1}, \delta^{3q}\}$. Then 
$$
f(X)=X^3h(X^{q-1}) = D_3X^{3q}+D_2X^{2q+1} +D_1X^{q+2}+D_0X^3
$$ is a PP over $\F_{q^2}$.

\end{thm}
\begin{proof}
From Lemma~\ref{L23}, we know that $f$ is a permutation of $\F_{q^2}$ if and only if $\gcd(3,q-1)=1$ and $g(X):=X^3h(X)^{q-1}$ permutes $\mu_{q+1}$. Since $q\equiv 2 \pmod 3$, the first condition holds trivially. Therefore it is sufficient to show that if $1+{\tilde \beta}\beta^3=0$ then $g$ permutes $\mu_{q+1}$. From Remark~\ref{R05}, we know that $h(X)=0$ has no solution in $\mu_{q+1}$. Now for any $\alpha \in \mu_{q+1}$, consider
 \[
  g(\alpha)= \alpha^3h(\alpha)^{q-1}=\alpha^3\frac{h(\alpha)^q}{h(\alpha)}= \frac{D_0^q\alpha^3+D_1^q\alpha^2 +D_2^q\alpha+D_3^q}{D_3\alpha^3+D_2\alpha^2 +D_1\alpha+D_0}.
 \]
Let $\displaystyle {\tilde g}(X):= \frac{D_0^qX^3+D_1^qX^2 +D_2^qX+D_3^q}{D_3X^3+D_2X^2 +D_1X+D_0} \in \F_{q^2}(X)$. It is easy to observe that $g$ permutes $\mu_{q+1}$ if and only if the rational function ${\tilde g}(X)$ permutes $\mu_{q+1}$. Consider the following system of equations
\begin{equation} \label{SEE1}
  \begin{cases}
   N_3 &= D_0^q,\\
   N_2 &= D_1^q,\\
   N_1 &= D_2^q,\\
   N_0 &= D_3^q,
  \end{cases}
  \iff
  \begin{cases}
   (\delta^3-{\tilde \delta}^q)({\tilde \beta} +\beta^{3q})=0 ,\\
   (\delta^{q+2}-{\tilde \delta}^q)(3{\tilde \beta}\beta +3\beta^{2q})=0,\\
   (\delta^{2q+1}-{\tilde \delta}^q)(3{\tilde \beta}\beta^2+3\beta^q)=0,\\
   (\delta^{3q}-{\tilde \delta}^q)(1+{\tilde \beta}\beta^3)  = 0,
  \end{cases}
 \end{equation}
 where $N_0, N_1, N_2$ and $N_3$ are as given in Equation~\eqref{C05}. If all the four conditions in the system~\eqref{SEE1} are satisfied then ${\tilde g}={\nu} \circ X^3 \circ \rho$. It is easy to verify that if $1+{\tilde \beta}\beta^3=0$ then  all the four conditions of the system~\eqref{SEE1} are satisfied. Thus, if $1+{\tilde \beta}\beta^3=0$ then ${\tilde g}$ is a permutation of $\mu_{q+1}$ and consequently $f$ is a permutation of $\F_{q^2}$.
\end{proof}


 \begin{exmp}
 Let $q=5$, $g$ be a generator of the cyclic group $\F_{q^2}^{*}$, $\beta=-1$, $\tilde{\beta}=1$ and $\delta=g=\tilde{\delta}$. From Theorem ~\ref{T05}, we obtain a permutation quadrinomial $f(X)=3(g+1)X^{15}+3gX^{11}+(g+1)X^7+2X^3$ of $\F_{q^2}$. In addition,
 let $q=5^3$, $g$ be a primitive element of $\F_{q^2}$ satisfying $g^6 + g^4 + g^3 + g^2 + 2=0$, $(\beta, \tilde{\beta}) = (-1, 1)$ and $(\delta, \tilde{\delta}) = (g^{14078}, g^{6470})$.
 Then from Theorem ~\ref{T05}, we get a permutation quadrinomial $f(X)= g^{12017}X^{3q} + g^{9477}X^{2q+1} + g^{10055}X^{q+2} + g^{7976}X^3$ of $\F_{q^2}$.
 \end{exmp}
We now present three more classes of permutation quadrinomials over $\F_{q^2}$. 
\begin{thm}\label{T15}
 Let $q\equiv 2 \pmod 3$ and $D_0, D_1, D_2, D_3$ are as given in the Equation~\eqref{C05}. Moreover, $\delta$ and $\tilde \delta$ in the expressions of $D_i$'s, $0 \leq i \leq 3$ are such that $\delta, \tilde \delta \in \F_{q^2}\setminus\F_{q}$ and $\tilde \delta \not \in \{ \delta^3, \delta^{q+2}, \delta^{2q+1}, \delta^{3q}\}$. If $1+{\tilde \beta}\beta^3=0$ for some ${\tilde \beta}\in \mu_{q+1}$ then the following quadrinomials are permutation quadrinomials over $\F_{q^2}$
 \begin{enumerate}
 \item $Xh_1(X^{q-1})$, where $h_1(X)= D_0X^q+D_3X^2 +D_2X+D_1$,
 \item  $X^{q}h_2(X^{q-1})$, where $h_2(X)=D_0X^{2q}+D_1X^q+D_3X+D_2$,
 \item  $X^{q-2}h_3(X^{q-1})$, where $h_3(X)=D_0X^{3q}+D_1X^{2q}+D_2X^q+D_3$.
 \end{enumerate}
\end{thm}
\begin{proof}
The proof follows along a similar line as in Theorem~\ref{T05}.
\end{proof}
\begin{exmp}
 Consider $q=2^3$ and $\F_{q^2}^{*}=\langle g \rangle $, where $g^6 + g^4 + g^3 + g + 1=0$. Let $\beta=1= \tilde{\beta}$ and $\delta=g=\tilde{\delta}$. Then Theorem ~\ref{T15} (1) implies that $f(X)=(g^5+g^3+1)X^{57}+(g^3+g)X^{15}+(g^5+g^4+g+1)X^8+(g^5+g^2)X$ is a permutation quadrinomial over $\F_{q^2}$. 
  \end{exmp}

\subsection*{Case 2:} When $f(X)$ is equivalent to $\zeta^{-1} \circ X^3 \circ \zeta,$~\mbox{where}~$q\equiv1 \pmod 3$ and for some $\delta \in \F_{q^2}\setminus \F_q$, $\zeta(X)=(X-\delta^q)/(X-\delta)$ and $ \zeta^{-1}(X)=(\delta X-\delta^q)/(X-1)$.
In this case, the function $\nu \circ f \circ \rho$ cannot be expressed in the form $X^r\frac{h(X)^q}{h(X)}$ for any $h(X)\in \F_{q^2}[X]$ such that $h(X)$ has no root in $\mu_{q+1}$ and  $X^r\frac{h(X)^q}{h(X)}$ permutes $\mu_{q+1}$. As a consequence, we can not construct PPs using our strategy in this case.
\subsection*{Case 3:}
 In this case, we shall construct two classes of permutation quadrinomials over $\F_{q^2}$ arising from degree-three permutation rational functions of $\F_q$, in the particular case of $q\equiv 0 \pmod 3$. From Lemma~\ref{T20}, we know that if $q\equiv 0 \pmod 3$ then any degree-three permutation rational function of $\bP^1(\F_q)$ will be equivalent to $X^3-\alpha X$ where either $\alpha =0$ or $\alpha$ is a non-square in $\F_q$. Let $\rho : \mu_{q+1} \rightarrow \bP^1(\F_q)$ and ${\nu} : \bP^1(\F_q) \rightarrow \mu_{q+1}$ are bijective degree-one rational functions. Then from Lemma~\ref{L21} and Lemma~\ref{L22}, we know that
 \[
  \rho(X) = \frac{\delta X -\beta \delta^q}{X-\beta},
 ~\mbox{and}~
  {\nu}(X) = {\tilde \beta} \left(\frac{X -{\tilde \delta}^q}{X- {\tilde \delta}}\right),
 \]
 for some $\beta, {\tilde \beta} \in \mu_{q+1}$ and $\delta, {\tilde \delta} \in \F_{q^2} \backslash \F_q$. It is straightforward to see that the rational function ${\nu} \circ (X^3 -\alpha X) \circ \rho : \mu_{q+1} \rightarrow \mu_{q+1}$ is a bijection of $\mu_{q+1}$. The rational function ${\nu} \circ (X^3 -\alpha X) \circ \rho$  can be written as ${\nu} \circ (X^3 -\alpha X) \circ \rho$ which is same as
 \begin{equation*}
  \begin{split}
  & \quad {\nu} \circ \left(\frac{(\delta X -\beta \delta^q)^3}{(X-\beta)^3} - \alpha \frac{(\delta X -\beta \delta^q)}{(X-\beta)} \right) \\
   &= {\nu} \circ \left(\frac{(\delta X -\beta \delta^q)^3-\alpha(\delta X -\beta \delta^q)(X-\beta)^2}{(X-\beta)^3} \right) \\
   &= {\nu} \circ \left(\frac{\delta^3 X^3 -\beta^3 \delta^{3q}-\alpha(\delta X -\beta \delta^q)(X^2+\beta X +\beta^2)}{X^3-\beta^3} \right) \\
   &= {\nu} \circ \left(\frac{(\delta^3-\alpha \delta) X^3+ \alpha \beta (\delta^q -\delta)X^2+\alpha \beta^2 (\delta^q -\delta)X +\beta^3 (\alpha \delta^q- \delta^{3q})}{X^3-\beta^3} \right) \\
   &= {\tilde \beta} \left(\frac{X -{\tilde \delta}^q}{X- {\tilde \delta}}\right) \circ \left(\frac{(\delta^3-\alpha \delta) X^3+ \alpha \beta (\delta^q -\delta)X^2+\alpha \beta^2 (\delta^q -\delta)X +\beta^3 (\alpha \delta^q- \delta^{3q})}{X^3-\beta^3} \right) \\
   &= {\tilde \beta} \left(\frac{(\delta^3-\alpha \delta - {\tilde \delta}^q) X^3+ \alpha \beta (\delta^q -\delta)X^2+\alpha \beta^2 (\delta^q -\delta)X +\beta^3 (\alpha \delta^q- \delta^{3q}+ {\tilde \delta}^q)}{(\delta^3-\alpha \delta -{\tilde \delta}) X^3+ \alpha \beta (\delta^q -\delta)X^2+\alpha \beta^2 (\delta^q -\delta)X +\beta^3 (\alpha \delta^q- \delta^{3q}+{\tilde \delta})}\right) \\
   &= \frac{N_3X^3+N_2X^2+N_1X+N_0}{D_3X^3+D_2X^2+D_1X+D_0},
  \end{split}
 \end{equation*}
 where
 \begin{equation}\label{C50}
  \begin{cases}
    N_0 &= {\tilde \beta}\beta^3 (\alpha \delta^q- \delta^{3q}+ {\tilde \delta}^q),\\
    N_1 &= \alpha {\tilde \beta} \beta^2 (\delta^q -\delta),\\
    N_2 &= \alpha {\tilde \beta} \beta (\delta^q -\delta),\\
    N_3 &= {\tilde \beta}(\delta^3-\alpha \delta - {\tilde \delta}^q),
  \end{cases}
  ~~~\mbox{and}~~~
  \begin{cases}
   D_0 &= \beta^3 (\alpha \delta^q- \delta^{3q}+ {\tilde \delta}),\\
    D_1 &= \alpha \beta^2 (\delta^q -\delta),\\
    D_2 &= \alpha \beta (\delta^q -\delta),\\
    D_3 &= \delta^3-\alpha \delta - {\tilde \delta},
  \end{cases}
 \end{equation}

 \begin{rmk}\label{R50}
  Notice that since $ {\nu} \circ X^3-\alpha X  \circ \rho: \mu_{q+1} \rightarrow \mu_{q+1}$ is a permutation of $\mu_{q+1}$, we have $D_3X^3+D_2X^2+D_1X+D_0 \neq 0$ for all $X \in \mu_{q+1}$.
 \end{rmk}

 Now we give the following lemma which will be used throughout the section.

 \begin{lem} \label{L550}
 Let
\begin{equation*}
 \begin{split}
  h(X)&=D_3X^3 +D_2X^2+D_1X+D_0,\\
  h_1(X) &= D_0X^q+D_3X^2 +D_2X+D_1,\\
  h_2(X) &=D_0X^{2q}+D_1X^q+D_3X+D_2,\\
  h_3(X) &=D_0X^{3q}+D_1X^{2q}+D_2X^q+D_3,
 \end{split}
\end{equation*}
where $D_0,D_1, D_2, D_3$ are defined in the Equation~\eqref{C50}. Then
 $h(X)=0$ has no solution in $\mu_{q+1}$ iff  $h_i(X)=0$ has no solution in $\mu_{q+1}$ for any $i=1,2,3$.
\end{lem}
\begin{proof}
 The proof follows along a similar line as Lemma~\ref{L05}. 
\end{proof}

We shall now use the rational function ${\nu} \circ (X^3-\alpha X) \circ \rho$ to construct permutation quadrinomials over $\F_{q^2}$. 
The proofs of the following theorems follow the same line as of Theorem~\ref{T05} and thus omitted.

\begin{thm}\label{T50} Let $q\equiv 0 \pmod 3$ and $h(X)=D_3X^3+D_2X^2 +D_1X+D_0$, where $D_0, D_1, D_2,$ $D_3$ are given in Equation~\eqref{C50}. Moreover, $\delta$ and $\tilde \delta$ in the expressions of $D_i$'s, $0 \leq i \leq 3$ are such that $\delta, \tilde \delta \in \F_{q^2}\setminus\F_{q}$ and $\tilde \delta \not \in \{ \delta^{3q}-\alpha \delta^q, \delta^3-\alpha \delta\}$, where $\alpha$ is either zero or a non-square in $\F_q$.
Let $\beta, {\tilde \beta}$ be elements in $\mu_{q+1}$ satisfying $1+{\tilde \beta}\beta^3=0$.
Then $f(X)=X^3h(X^{q-1})$ is a PP over $\F_{q^2}$.
\end{thm}

 \begin{exmp}
 In the above Theorem~\ref{T50}, if we choose $\beta=2$, $\tilde{\beta}=1$, $\delta=g=\tilde{\delta}$ and $\alpha=2$. Then we get $f(X)=(g+1)X^7+2(g+1)X^{5}+(2g+1)X^3+(2g+1)X$, a permutation polynomial over the finite field $\F_{q^2}$, where $q=3$ and $g$ is generator of the cyclic group $\F_{q^2}^{*}$. In addition, by taking $q=3^4$, a primitive element $g$ of $\F_{q^2}$ satisfying $g^8 + 2g^5 + g^4 + 2g^2 + 2g + 2=0$, $\alpha=g^{q+1}$, $(\beta, \beta_1) = (-1, 1)$ and $(\delta, \tilde{\delta}) = (g^{4898}, g^{332})$,
 we can obtain a permutation quadrinomial $f(X) = g^{1523}X^{3q} + g^{1681}X^{2q+1} + g^{4961}X^{q+2} + g^{3736}X^3$ over $\F_{q^2}$.
\end{exmp}


\begin{thm}\label{T51}
  Let $q\equiv 0 \pmod 3$ and $D_0, D_1, D_2,$ $D_3$ are given in Equation~\eqref{C50}. Moreover, $\delta$ and $\tilde \delta$ in the expressions of $D_i$'s, $0 \leq i \leq 3$ are such that $\delta, \tilde \delta \in \F_{q^2}\setminus\F_{q}$ and $\tilde \delta \not \in \{ \delta^{3q}-\alpha \delta^q, \delta^3-\alpha \delta\}$, where $\alpha$ is either zero or a non-square in $\F_q$. If $1+{\tilde \beta}\beta^3=0$ for some ${\tilde \beta}\in \mu_{q+1}$ then the following quadrinomials are permutation quadrinomials over $\F_{q^2}$
 \begin{enumerate}
 \item $Xh_1(X^{q-1})$, where $h_1(X)= D_0X^q+D_3X^2 +D_2X+D_1$,
 \item  $X^{q}h_2(X^{q-1})$, where $h_2(X)=D_0X^{2q}+D_1X^q+D_3X+D_2$,
 \item  $X^{3q}h_3(X^{q-1})$, where $h_3(X)=D_0X^{3q}+D_1X^{2q}+D_2X^q+D_3$.
 \end{enumerate}
\end{thm}
 \begin{exmp}
 Let $q=3^2$ and $g$ be a generator of the cyclic group $\F_{q^2}^{*}$ satisfying $X^4 + 2X^3 + 2=0$. Consider $\beta=1$, $\tilde{\beta}=-1$, $\delta=g=\tilde{\delta}$ and $\alpha=g^{q+1}$. Then we obtain $f(X)=(g+1)X^{27}+(g^3+g+1)X^{19}+(g^3+g+1)X^{11}+(g^3+g^2+2)X^3$, a permutation quadrinomial over the finite field $\F_{q^2}$ from Theorem~\ref{T51} (3).
\end{exmp}

\begin{rmk}
For $q \equiv 2 \pmod 3$ or $q \equiv 0 \pmod 3$, it follows from Remark~\ref{R1} that any class of permutation quadrinomials which can be obtained using any degree three permutation rational function is already studied either in Case 1 or Case 3 discussed above. Recently, {\"O}zbudak and Tem{\"u}r~\cite{FB_quad_2023} classified all permutation polynomials over $\F_{q^2}$ of the form $f(X) = X^3 + aX^{q+2} +bX^{2q+1} + cX^{3q}$, where $a, b, c \in \F_{q}^*$ and $q$ is an odd prime power. Notice that the quadrinomials obtained in~\cite{FB_quad_2023} permutes $\F_{q^2}$ only if $q \equiv 2 \pmod 3$ or $q \equiv 0 \pmod 3$. Hence the classes of permutation quadrinomials characterized in~\cite{FB_quad_2023} are included in our proposed classes.
It will be shown in Section \ref{S5} that there exist new permutation quadrinomials that are not QM-equivalent to all known ones.
\end{rmk}


\section{Permutation polynomials from degree-four rational functions} \label{S4}
In this section, we study six classes of permuation pentanomials and two classes of permutation  binomials over the finite field $\F_{q^2}$ in even characteristic, using degree-four rational functions given in the first case of Lemma~\ref{T21}.

Let $q$ be even and $f(X) \in \F_q(X)$ is equivalent to $X^4+bX^2+aX$ for some $a,b \in \F_q$ such that $X^3+bX+a$ has no roots in $\F_q^*$. This degree-four function is of the form of linearized polynomials over $\F_q$ in even characteristic.

First, we would like to review some background material on linearized polynomials over  $\F_{q}$, that would be used to characterise degree-four rational functions $f(X)$ permuting $\bP^1(\F_{q})$.
Polynomials over $\F_q$ of the form $ L(X)=\sum_{i=0}^{m-1} a_iX^{2^i} \in \F_q[X]$ are often known as additive polynomials or linearized polynomials. Such a special kind of polynomials can induce linear transformations of vector space $\F_2^m$ over $\F_2$. Linearized polynomials are interesting objects especially when they are bijective. More precisely, a linearized polynomial $L(X)$ permutes the finite field $\F_q$ if and only if its only root in $\F_q$ is zero. 
The degree-four linearized polynomials $X^4+bX^2+aX \in \F_q[X]$ permutes $\F_q$ if and only if $P_{a,b}(X):=X^3+bX+a=0$ has no solution in $\F_q$. When $(a,b)=(0,0)$ it is clear that $X^4$ permutes $\F_q$; when $a\neq0$ and $b=0$, the polynomial $P_{a,b}=X^4+aX$ is a permutation of $\F_q$ if and only if $a$ is a non-cube element in $\F_q$; and lastly, when $b \neq0$, there exists a unique $c \in \F_q^*$ such that $b=c^2$ and by a simple substitution $X=cX$, we can transform $P_{a,b}(X)$ into the
form $P_{\alpha}(X)=X^3+X+\alpha$, where $\alpha=\frac{a}{c^3}$. Thus the linearized polynomial $X^4+bX^2+aX$ is a permutation polynomial if and only if $P_{\alpha}(X)$ is irreducible over $\F_q$.
The following lemma gives necessary and sufficient conditions for polynomial $P_{\alpha}(X)$ to be irreducible over $\F_{q^2}$.
\begin{lem}\cite{Wil} \label{L49}
Let $m$ be a positive integer and $P_{a,b}(X)=X^3+aX+b\in \F_{2^m}[X]$ be a polynomial, where $b \in \F_{2^m}^{*}$. Then the polynomial $f(X)$ is irreducible over $\F_{2^m}$ if and only if $\Tr\left(\dfrac{a^3}{b^2}\right)=\Tr(1)$, and $t_1, t_2$ are not cubes in $\F_{2^m}$($m$ even), $\F_{2^{2m}}$ ($m$ odd), where $t_1$ and $t_2$ are roots of the equation $t^2+bt+a^3$.
\end{lem}

Bracken, Tan and Tan in~\cite{Bracken_14} further simplified the above conditions and gave the following result for even positive integer $m$.

\begin{lem}\cite{Bracken_14} \label{L50}
Let $m=2k$. Then the polynomial $P_\alpha(X)=X^3+X+ \alpha$ is irreducible over $\F_{2^m}$ if and only if $\alpha=a+a^{-1}$ for some non-cube $a \in \F_{2^m}$.
\end{lem}
From Lemma~\ref{L49} and Lemma~\ref{L50}, we know that for even $q \geq 8$, the degree-four rational function $f(X)=X^4 + bX^2 + aX$ permutes $\bP^1(\F_{q})$ will be equivalent to either of the following.
\begin{enumerate}[(i)]
 \item $X^4+X^2+\alpha X$, where $X^3+X+\alpha$ is irreducible over $\F_{2^n}$;
 \item $X^4+aX$, $a$ is some non-cube element in $ \F_{q}^*$;
 \item $X^4$.
\end{enumerate}
In what follows, we shall deal each of the above subcases individually.

\subsubsection*{\textbf{Subcase 1.1:}}We shall now use the linearized polynomial  $f(X)=X^4+X^2+\alpha X$ such that $X^3+X+\alpha$ is irreducible over $\F_q$, to construct rational functions that induce a bijection of $\mu_{q+1}$, the unit circle of $\F_{q^2}$ with order $q+1$. Let $\rho \in \F_{q^2}(X)$ be a degree-one rational function which induces a bijection from $\mu_{q+1}$ to $\bP^1(\F_{q})$. Then from Lemma~\ref{L21}, we know that
\[
  \rho(X) = \frac{\delta X +\beta \delta^{q}}{X+\beta}
\]
for some $\beta \in \mu_{q+1}$ and $\delta \in \F_{q^2}\backslash \F_{q}$. Similarly, let $\nu \in \F_{q^2}(X)$ be a degree-one rational function which induces bijection from $\bP^1(\F_{q})$ to $\mu_{q+1}$. Then from Lemma~\ref{L22}, we know that
\[
 \nu(X)= {\tilde \beta} \left(\frac{X +{\tilde \delta}^{q}}{X+ {\tilde \delta}}\right),
\]
for some ${\tilde \beta} \in \mu_{q+1}$ and ${\tilde \delta} \in \F_{q^2} \backslash \F_{q}$. Now consider the rational function $\nu \circ f \circ \rho : \mu_{q+1} \rightarrow \mu_{q+1}$ which permutes the set $\mu_{q+1}$. The composition $ \nu \circ f \circ \rho ~(X)$ is given by
 \begin{equation*}
  \begin{split}
  & \nu \circ \left(\frac{(\delta X +\beta \delta^q)^4}{(X+\beta)^4}+\frac{(\delta X +\beta \delta^q)^2}{(X+\beta)^2}+\frac{\alpha(\delta X +\beta \delta^q)}{(X+\beta)} \right) \\
   &=  {\tilde \beta} \left(\frac{X +{\tilde \delta}^q}{X+ {\tilde \delta}}\right) \circ \left(\frac{(\delta X +\beta \delta^q)^4 + (\delta X +\beta \delta^q)^2 (X+\beta)^2 + \alpha(\delta X +\beta \delta^q) (X+\beta)^3 }{X^4+\beta^4}\right) \\
   &= {\tilde \beta} \left(\frac{X +{\tilde \delta}^q}{X+ {\tilde \delta}}\right) \circ \left(\frac{N_4X^4+N_3X^3+N_2X^2+N_1X+N_0 }{X^4+\beta^4}\right) \\
   &= {\tilde \beta} \frac{(N_4+{\tilde \delta}^q) X^4+N_3X^3+N_2X^2+N_1X+(N_0 + \beta^4{\tilde \delta}^q) }{(N_4+{\tilde \delta}) X^4+N_3X^3+N_2X^2+N_1X+(N_0 + \beta^4{\tilde \delta}) },
  \end{split}
 \end{equation*}
 where
 \begin{equation}\label{C60}
  \begin{cases}
   N_0 &= \beta^4(\delta^{4q}+\delta^{2q}+\alpha\delta^q),\\
   N_1 &= \alpha\beta^3(\delta+\delta^q),\\
   N_2 &= \beta^2(\delta+\delta^q)(\delta+\delta^q+\alpha),\\
   N_3 &= \alpha\beta(\delta+\delta^q),\\
   N_4 &= \delta^4+\delta^2+\alpha\delta.
  \end{cases}
 \end{equation}


We shall now prove the following lemma which will be used in the construction of rational functions that permute the set $\mu_{q+1}$.

\begin{lem} \label{L61}
 Let
 \begin{equation*}
  \begin{split}
   h_1(X)&= (N_4+{\tilde \delta})X^4 +N_3X^3+ N_2X^2+ N_1X+(N_0 + \beta^4{\tilde \delta}),\\
   h_2(X)&= (N_0 +\beta^4  {\tilde \delta})X^q+(N_4+{\tilde \delta})X^3 +N_3X^2+ N_2X+ N_1,\\
   h_3(X)&= (N_0 + \beta^4 {\tilde \delta})X^{2q}+N_1X^q+(N_4+{\tilde \delta})X^2 +N_3X+ N_2, \\
   h_4(X) &= (N_0 +\beta^4 {\tilde \delta})X^{3q}+N_1X^{2q}+N_2X^q+(N_4+{\tilde \delta})X +N_3, \\
   h_5(X) &= (N_0 + \beta^4{\tilde \delta})X^{4q}+N_1X^{3q}+N_2X^{2q}+N_3X^q+(N_4+{\tilde \delta}),
  \end{split}
 \end{equation*}
where $N_0, N_1, N_2, N_3$ and $N_4$ are as given in system~\eqref{C60}. 
Then the polynomial $h_i(X)=0$ for $i=1,2,3,4,5$ has no solution $X \in \mu_{q+1}$.
\end{lem}
\begin{proof} Since the composition $\nu \circ f \circ \rho$ is a bijection of $\mu_{q+1}$,
by \eqref{C60}, the polynomial 
$$
h_1(X)= (N_4+{\tilde \delta})X^4 +N_3X^3+ N_2X^2+ N_1X+(N_0 + \beta^4{\tilde \delta}) = 0
$$ has no solution in $\mu_{q+1}$. Furthermore, for any $z\in \mu_{q+1}$,
we have 
$$
h_2(z) = h_1(z)/z, h_3(z) = h_1(z)/z^2, h_4(z) = h_1(z)/z^3, h_5(z) = h_1(z)/z^4.
$$
The desired conclusion thus follows.
\end{proof}

Now using Lemma~\ref{L61} and the rational function $\nu \circ f \circ \rho$, we shall construct three classes of permutation pentanomials over $\F_{q^2}$.

\begin{thm}\label{T60}Assume $X^3+X+\alpha$ is irreducible over $\F_q$. 
Let $\delta, \tilde \delta \in \F_{q^2}\setminus\F_{q}$ such that $\tilde \delta \not \in \{ \delta^4+\delta^2+\alpha \delta, \delta^{4q}+\delta^{2q}+\alpha \delta^q\}$, $\delta+\delta^q+\alpha \neq 0$.
Assume $\beta, \tilde{\beta}$ in $\mu_{q+1}$ satisfies ${\tilde \beta}\beta^4=1$. Then,
the pentanomial
$f_1(X)=X^{4}h_1(X^{q-1})$ for 
$h_1(X)=(N_4+{\tilde \delta}) X^4+N_3X^3+N_2X^2+N_1X+(N_0 + \beta^4{\tilde \delta})$
is a PP over $\F_{q^2}$.
\end{thm}
\begin{proof}
From Lemma~\ref{L23}, we know that $f_1$ is a permutation polynomial over $\F_{q^2}$ if and only if $\gcd(r_1,q-1)=1$ and $g_1(X):=X^{r_1}h_1(X)^{q-1}$ permutes $\mu_{q+1}$. It is clear that $\gcd(4, q-1)=1$. Therefore, it is sufficient to show that if ${\tilde \beta}\beta^4=1$ then $g_1$ permutes $\mu_{q+1}$. Note that $h_1(X)=0$ has no solution in $\mu_{q+1}$ when $\tilde \delta \not \in \{ \delta^4+\delta^2+\alpha \delta, \delta^{4q}+\delta^{2q}+\alpha \delta^q\}$,  and $\delta+\delta^q+\alpha \neq 0$. Now for any $z \in \mu_{q+1}$, consider
 \[
  {\tilde g_1}(z):=z^4\frac{h_1(z)^{q}}{h_1(z)}= \frac{(N_0 + \beta^4{\tilde \delta})^{q}z^4+N_1^{q}z^3 +N_2^{q}z^2+N_3^{q}z+(N_4+{\tilde \delta})^{q}}{(N_4+{\tilde \delta}) z^4+N_3z^3+N_2z^2+N_1z+(N_0 + \beta^4{\tilde \delta})}.
 \]
It is easy to observe that $g_1$ is a bijection of $\mu_{q+1}$ if and only if $X \mapsto {\tilde g_1}(X)$ is a bijection of $\mu_{q+1}$. 
Recall that 
$$
\nu \circ f \circ \rho ~(X) = {\tilde \beta} \frac{(N_4+{\tilde \delta}^q) X^4+N_3X^3+N_2X^2+N_1X+(N_0 + \beta^4{\tilde \delta}^q) }{(N_4+{\tilde \delta}) X^4+N_3X^3+N_2X^2+N_1X+(N_0 + \beta^4{\tilde \delta}) }.
$$
It can be verified via a routine calculation that if ${\tilde \beta}\beta^4=1$ then ${\tilde g_1}= \nu \circ f \circ \rho$. This completes the proof.
\end{proof}

\begin{exmp}\label{EX41}
For $q=4$, $f(X)=(b^2+b+1)X^{13}+X^{10}+(b^2+b+1)X^7+(b^2+1)X^4+(b+1)X$ is a permutation polynomial over the finite field $\F_{q^2}$, where $b$ is generator of the cyclic group $\F_{q^2}^{*}$. We have obtained PP by choosing $\beta=1$, $\delta=b^3=\tilde{\delta}$ and $\alpha=1$ for the Theorem \ref{T60}. In addition, take $q=2^8$, $\beta=1$, $\alpha = b^{q+1} + b^{(q+1)(q-2)}$ and $(\delta, \tilde{\delta})=(b^{321}, b^{47351})$, where $b$ is a primitive element of $\F_{q^2}$ satisfying $b^{16}+b^5+b^3+b^2+1=0$.
Then we obtain from Theorem \ref{T60} a permutation $f(X) = b^{8722}X^{4q} +b^{48830}X^{3q+1} + b^{53713}X^{2q+2} + b^{48830}X^{q+3} + b^{47311}X^4$ over $\F_{q^2}$.
\end{exmp}

\begin{thm}\label{T61}
 Let $h_2(X)=(N_0 +\beta^4  {\tilde \delta})X^{q}+(N_4+{\tilde \delta})X^3 +N_3X^2+ N_2X+ N_1,$ where $N_0$, $N_1,$ $N_2,$ $N_3$ and $N_4$ are as given in system~\eqref{C60} and $r_2=2$. Moreover, $\delta$ and $\tilde \delta$ in the expressions of $N_i$'s, $0 \leq i \leq 4$ are such that $\delta, \tilde \delta \in \F_{q^2}\setminus\F_{q}$, $\tilde \delta \not \in \{ \delta^4+\delta^2+\alpha \delta, \delta^{4q}+\delta^{2q}+\alpha \delta^q\}$, $\delta+\delta^q+\alpha \neq 0$ and $X^3+X+\alpha$ is irreducible over $\F_q$.
If ${\tilde \beta}\beta^4=1$ for some ${\tilde \beta} \in \mu_{q+1}$,
Then $f_2(X)=X^{r_2}h_2(X^{q-1})$ is a PP over $\F_{q^2}$.
\end{thm}
\begin{proof}
Since $\gcd(2, q-1)$ is always equal to one, in view of Lemma~\ref{L23}, we only need to show that if ${\tilde \beta}\beta^4=1$ then $g_2(X)=X^2h_2(X)^{q-1}$ permutes $\mu_{q+1}$. From Lemma~\ref{L61}, we know that $h_2(X)=0$ has no solution in $\mu_{q+1}$ and hence $g_2$ permutes $\mu_{q+1}$ if and only if the rational function
\[
 {\tilde g_2}(X):= X^2 \frac{h_2^{(q)} \left( \frac{1}{X}\right)}{h_2(X)}
\]
where $h_2^{(q)}(X)$ is the polynomial obtained from $h_2(X)$ by raising all coefficients to the $q$-th power, permutes $\mu_{q+1}$. A routine calculation shows that if ${\tilde \beta}\beta^4=1$ then ${\tilde g_2}(z) = \nu \circ f \circ \rho (z)$ for all $z \in \mu_{q+1}$. Thus, if ${\tilde \beta}\beta^4=1$ then ${\tilde g_2}$ is a permutation of $\mu_{q+1}$. This completes the proof.
\end{proof}

\begin{rmk}
 In the construction of permutation polynomials of the form $f(X)=X^rh(X^{q-1})$ over finite field $\F_{q^2}$, the exponents in the polynomial $h(X)$ can be viewed as modulo $(q+1)$.
%
 We obtained the same class of permutation pentanomials from $h_3(X)$ and $h_4(X)$, So we are not including results from these two functions.
\end{rmk}

\begin{thm}\label{T64}
 Let $h_5(X) = (N_0 + \beta^4{\tilde \delta})X^{q-3}+N_1X^{q-2}+N_2X^{q-1}+N_3X^q+(N_4+{\tilde \delta}),$ where $N_0$, $N_1,$ $N_2,$ $N_3$ and $N_4$ are as given in system~\eqref{C60} and $r_5=q-3$. Moreover, $\delta$ and $\tilde \delta$ in the expressions of $N_i$'s, $0 \leq i \leq 4$ are such that $\delta, \tilde \delta \in \F_{q^2}\setminus\F_{q}$, $\tilde \delta \not \in \{ \delta^4+\delta^2+\alpha \delta, \delta^{4q}+\delta^{2q}+\alpha \delta^q\}$, $\delta+\delta^q+\alpha \neq 0$ and $X^3+X+\alpha$ is irreducible over $\F_q$. Then $f_5(X)=X^{r_5}h_5(X^{q-1})$ is a PP over $\F_{q^2}$ when the elements $\beta, \tilde{\beta} \in \mu_{q+1}$ satisfy ${\tilde \beta}\beta^4+1=0$.
\end{thm}
\begin{proof}
The proof follows by using a similar argument as in Theorem~\ref{T61}.
\end{proof}

\subsubsection*{\textbf{Subcase 1.2:}}\label{SS60}
In this subcase, we take the case where $f(X)=X^4+a X$, where $a$ is nonzero and non-cube element in $\F_q$. Now consider the rational function $\nu \circ f \circ \rho : \mu_{q+1} \rightarrow \mu_{q+1}$ which permutes the set $\mu_{q+1}$ and is given by $\nu \circ f \circ \rho(X)$, which is equal to,
\begin{align*}
& \Tilde{\beta}\frac{X+{\Tilde{\delta}}^q}{X+\Tilde{\delta}} \circ (X^4+a X) \circ \frac{\delta X+\beta \delta^q}{X+\beta}\\
= & \Tilde{\beta}\frac{X+{\Tilde{\delta}}^q}{X+\Tilde{\delta}} \circ \left( \left(\frac{\delta X+\beta \delta^q}{X+\beta}\right)^4+a \left(\frac{\delta X+\beta \delta^q}{X+\beta}\right)\right)\\
= & \Tilde{\beta}\frac{\left(\frac{\delta X+\beta \delta^q}{X+\beta}\right)^4+a \left(\frac{\delta X+\beta \delta^q}{X+\beta}\right)+{\Tilde{\delta}}^q}{ \left(\frac{\delta X+\beta \delta^q}{X+\beta}\right)^4+a \left(\frac{\delta X+\beta \delta^q}{X+\beta}\right)+\Tilde{\delta}}\\
= & \Tilde{\beta} \frac{(\delta^4+a\delta+{\Tilde{\delta}}^q)X^4+a(\delta+\delta^q)\left(\beta X^3+\beta^2X^2+\beta^3 X\right)+\beta^4(a\delta^q+{\Tilde{\delta}}^q+\delta^{4q})}{(\delta^4+a\delta+{\Tilde{\delta}})X^4+a(\delta+\delta^q)\left(\beta X^3+\beta^2X^2+\beta^3X\right)+\beta^4(a\delta^q+{\Tilde{\delta}}+\delta^{4q})}\\
= & \Tilde{\beta} \frac{(N_4+{\Tilde{\delta}}^q)X^4+N_3 X^3+N_2 X^2+N_1 X+(N_0+\beta^4 {\Tilde{\delta}}^q)} {(N_4+{\Tilde{\delta}})X^4+N_3 X^3+N_2 X^2+N_1 X+(N_0+\beta^4 {\Tilde{\delta}})}\\
\end{align*}
where
 \begin{equation}\label{C61}
  \begin{cases}
   N_0 &= \beta^4(a\delta^q+\delta^{4q}),\\
   N_1 &= a\beta^3(\delta+\delta^q),\\
   N_2 &= a\beta^2(\delta+\delta^q),\\
   N_3 &= a\beta(\delta+\delta^q),\\
   N_4 &= \delta^4+a\delta.
  \end{cases}
 \end{equation}
  \begin{rmk}\label{R61}
Let $N_0, N_1, N_2, N_3$ and $N_4$ are as given in system~\eqref{C61}. Then the equation
$$(N_4+{\tilde \delta})X^4 +N_3X^3+ N_2X^2+ N_1X+(N_0 + \beta^4{\tilde \delta}) = 0$$ has no solution $X \in \mu_{q+1}$, as $\nu \circ f \circ \rho$ is a bijection of $\mu_{q+1}$.
 \end{rmk}
 
 Next, we give the following lemma which will be used in constructing rational functions that permutes the set $\mu_{q+1}$.

\begin{lem} \label{L62}
 Let
 \begin{equation*}
  \begin{split}
   h_1(X)&= (N_4+{\tilde \delta})X^4 +N_3X^3+ N_2X^2+ N_1X+(N_0 + \beta^4{\tilde \delta}),\\
   h_2(X)&= (N_0 +\beta^4  {\tilde \delta})X^q+(N_4+{\tilde \delta})X^3 +N_3X^2+ N_2X+ N_1,\\
   h_3(X)&= (N_0 + \beta^4 {\tilde \delta})X^{2q}+N_1X^q+(N_4+{\tilde \delta})X^2 +N_3X+ N_2, \\
   h_4(X) &= (N_0 +\beta^4 {\tilde \delta})X^{3q}+N_1X^{2q}+N_2X^q+(N_4+{\tilde \delta})X +N_3, \\
   h_5(X) &= (N_0 + \beta^4{\tilde \delta})X^{4q}+N_1X^{3q}+N_2X^{2q}+N_3X^q+(N_4+{\tilde \delta}),
  \end{split}
 \end{equation*}
where $N_0, N_1, N_2, N_3$ and $N_4$ are as given in system~\eqref{C61}. Then the following statements are equivalent:
\begin{enumerate}[(a)]
 \item $h_1(X)=0$ has no solution $X \in \mu_{q+1}$,
 \item $h_2(X)=0$ has no solution $X \in \mu_{q+1}$,
 \item $h_3(X)=0$ has no solution $X \in \mu_{q+1}$,
 \item $h_4(X)=0$ has no solution $X \in \mu_{q+1}$,
 \item $h_5(X)=0$ has no solution $X \in \mu_{q+1}$.
\end{enumerate}
\end{lem}

\begin{proof} 
The proof follows by using a similar argument as in Lemma~\ref{L61}.
\end{proof}

We shall now construct three classes of permutation pentanomials over finite fields $\F_{q^2}$ using Lemma~\ref{L62} and the rational function $\nu \circ f \circ \rho$, .
\begin{thm}\label{T65}
Let $h_1(X)=(N_4+{\tilde \delta}) X^4+N_3X^3+N_2X^2+N_1X+(N_0 + \beta^4{\tilde \delta})$, where $N_0$, $N_1,$ $N_2,$ $N_3$ and $N_4$ are as given in system~\eqref{C61} and $r_1=4$. Moreover, $\delta$ and $\tilde \delta$ in the expressions of $N_i$'s, $0 \leq i \leq 4$ are such that $\delta, \tilde \delta \in \F_{q^2}\setminus\F_{q}$ and $\tilde \delta \not \in \{ \delta^4+a \delta, \delta^{4q}+a \delta^q\}$ where $a$ is nonzero and non-cube element in $\F_q$.
Then $f_1(X)=X^{r_1}h_1(X^{q-1})$ is a PP over $\F_{q^2}$ when the elements $\beta, \tilde{\beta} \in \mu_{q+1}$ satisfy ${\tilde \beta}\beta^4+1=0$.
\end{thm}

\begin{proof} Using the techniques similar to that of Theorem ~\ref{T60} yields the proof.
\end{proof}

\begin{exmp}\label{EX42}
Let $b$ be a primitive element of $\F_{q^2}$ with $q=4$. Take $\beta=b^3=\tilde{\beta}$, $\delta=b=\tilde{\delta}$ and $a=b^2+b$. From the above Theorem \ref{T65}
we get $f(X)=(b^2+1)X^{13}+(b^3+b^2+b)X^{10}+(b^3+1)X^7+(b^3+b^2)X^4+(b^3+b^2+1)X$ that is a permutation polynomial over $\F_{q^2}$. Furthermore, take $q=2^6$, $\beta=b^{3087}$, $a = b^{q+1}$, $(\delta, \tilde{\delta})=(b^{3894}, b^{3990})$, where $b$ is a primitive element of $\F_{q^2}$ satisfying $b^{12}+b^7+b^6+b^5+b^3+b+1=0$.
Then we obtain from Theorem \ref{T65} a permutation $f(X) = b^{28}X^{4q} +b^{3477}X^{3q+1} + b^{2469}X^{2q+2} + b^{1461}X^{q+3} + b^{3107}X^4$ over $\F_{q^2}$.
\end{exmp}

\begin{thm}\label{T66}
 Let $h_2(X)=(N_0 +\beta^4  {\tilde \delta})X^{q}+(N_4+{\tilde \delta})X^3 +N_3X^2+ N_2X+ N_1,$ where $N_0$, $N_1,$ $N_2,$ $N_3$ and $N_4$ are as given in system~\eqref{C61} and $r_2=2$. Moreover, $\delta$ and $\tilde \delta$ in the expressions of $N_i$'s, $0 \leq i \leq 4$ are such that $\delta, \tilde \delta \in \F_{q^2}\setminus\F_{q}$ and $\tilde \delta \not \in \{ \delta^4+a \delta, \delta^{4q}+a \delta^q\}$ where $a$ is nonzero and non-cube element in $\F_q$. If ${\tilde \beta}\beta^4=1$ for some ${\tilde \beta} \in \mu_{q+1}$, then $f_2(X)=X^{r_2}h_2(X^{q-1})$ is a PP over $\F_{q^2}$.
\end{thm}

\begin{proof}
An analogous argument adopted in Theorem ~\ref{T61} leads directly to the proof.
\end{proof}

\begin{rmk}
 In constructing permutation polynomials of the form $f(X)=X^rh(X^{q-1})$ over finite field $\F_{q^2}$, the exponents in the polynomial $h(X)$ can be viewed as modulo $(q+1)$.
%
 We are not getting new classes of permutation polynomials from $h_3(x)$ and $h_4(x)$, so we are not including  those results here.
\end{rmk}

Similarly we have the following theorem, 
\begin{thm}\label{T69}
 Let $h_5(X) = (N_0 + \beta^4{\tilde \delta})X^{q-3}+N_1X^{q-2}+N_2X^{q-1}+N_3X^q+(N_4+{\tilde \delta}),$ where $N_0$, $N_1,$ $N_2,$ $N_3$ and $N_4$ are as given in system~\eqref{C61} and $r_5=q-3$. Moreover, $\delta$ and $\tilde \delta$ in the expressions of $N_i$'s, $0 \leq i \leq 4$ are such that $\delta, \tilde \delta \in \F_{q^2}\setminus\F_{q}$ and $\tilde \delta \not \in \{ \delta^4+a \delta, \delta^{4q}+a \delta^q\}$ where $a$ is nonzero and non-cube element in $\F_q$. Then $f_5(X)=X^{r_5}h_5(X^{q-1})$ is a PP over $\F_{q^2}$ when the elements $\beta, \tilde{\beta} \in \mu_{q+1}$ satisfy ${\tilde \beta}\beta^4+1=0$.
\end{thm}

\begin{rmk}
For even $q$, Remark~\ref{R1} implies that any class of permutation pentanomials derived from a degree-four permutation rational function has already been studied either in Subcase 1.1 or Subcase 1.2 discussed above. Recently, Rai and Gupta~\cite[Theorem 3.3]{Rai} characterized a class of permutation pentanomials over $\F_{q^2}$ with its coefficients in $\F_q$, where $q=2^m$ and $m \geq 5$. Their result~\cite{Rai}  was obtained by determining the permutation behavior of degree four rational function. Consequently, the permutation pentanomials they identified form a subclass of those introduced in this paper.
\end{rmk}

\subsubsection*{\textbf{Subcase 1.3:}} In this subcase, we consider  $f(X)=X^4$ and give two classes of permutation binomials by putting $a=0$ in the above Subcase 1.2. Using same techniques as above, we summarise the permutation binomials together in the following theorem.
\begin{thm}\label{T70} For $\beta \in \mu_{q+1}$ satisfying ${\tilde \beta}\beta^4=1$, $\delta, \tilde \delta \in \F_{q^2}\backslash \F_{q}$ and $\tilde \delta \not \in \{\delta^4, \delta^{4q}\}$, we have the following
 \begin{enumerate}
\item $f_2(X)=X^2 h_2(X^{q-1})$, where $h_2(X)= (\beta^4 \delta^{4q} +\beta^4  {\tilde \delta})X^q+(\delta^4+{\tilde \delta})X^3$, is a permutation binomial over $\F_{q^2}$.

\item $f_5(X)=X^{q-3}h_5(X^{q-1})$, where $h_5(X)=(\beta^4 \delta^{4q} + \beta^4{\tilde \delta})X^{4q}+(\delta^4+{\tilde \delta})$, is a permutation binomial over $\F_{q^2}$.
\end{enumerate}
\end{thm}

\begin{rmk}
 In the above theorem, we are only considering the PPs corresponding to $h_2(x)$ and $h_5(x)$ only. This is because the polynomials obtained from $h_3(x)$ and $h_4(x)$ are same as the class obtained from $h_2(x)$ and the polynomials obtained from $h_1(x)$ are linearized polynomials over $\F_{2^n}$.
\end{rmk}

\begin{exmp}\label{EX43}
Let $q=4$ and $b$ be a primitive element of $\F_{q^2}$. By choosing $\beta=b^3=\tilde{\beta}$, $\delta=b$ and $\tilde{\delta}=b^3$ in the Theorem \ref{T70},
we obtain a permutation polynomial
$f(X)=(b^3+b^2)X^{14}+(b^4+b^3)X^{11}$ over $\F_{q^2}$. Moreover, take $q=2^8$, $\beta=b^{31110}$, $(\delta, \tilde{\delta})=(b^{53660}, b^{33334})$, where $b$ is a primitive element of $\F_{q^2}$ satisfying $b^{16}+b^5+b^3+b^2+1=0$.
Then we obtain from Theorem \ref{T65} a permutation $f(X) = b^{34047}X^{q^2-q+2} +b^{53717}X^{3q-1}$ over $\F_{q^2}$.
\end{exmp}


Note that Lemma~\ref{T21} consists of two additional cases of degree-four permutation rational functions as characterized in~\cite[Theorem~1.4]{Ding2020}. However, these cases are more computationally demanding, primarily due to the requirement that the composition $\nu \circ f \circ \rho$ must take the form $X^r h(X)^{q-1}$ for some $h(X) \in \mathbb{F}_{q^2}[X]$. The development of more effective methods to handle these cases remains an open problem and a promising direction for further research.

\section{Quasi-multiplicative equivalence}\label{S5}

In this section, we discuss the QM equivalence of our proposed permutation polynomials with the known ones. Throughout the section, $\mathbb{Z}_{q^2-1}$ denotes the ring of integers modulo $q^2-1$.

We first determine the QM inequivalence of the obtained permutation quadrinomials in Section~\ref{S3} with the known classes of permutation quadrinomials listed in Table~\ref{Table0} in Appendix~\ref{quad}. Notice that it is sufficient to show the QM inequivalence of the obtained permutation quadrinomials with the known classes of permutation quadrinomials over finite fields of odd characteristic only, as it would confirm that our obtained classes of permutation quadrinomials are new.

\begin{thm}\label{quadeq1}
Let $f(X)$ be as in Theorem~\ref{T05} then $f(X)$ is QM inequivalent to the all permutation quadrinomials $F_i$ listed in Table~\ref{Table0} for $1 \leq i \leq 19$ over $\F_{q^2}$.
\end{thm}
\begin{proof}
Notice that $F_i$ for $i \in \{1,2,3,4,5,6,7,8,12,13,14 \}$ are permutations for $p=3$. Since $q \equiv 2 \pmod 3$ for $f(X)$ to be a permutation, there is no need to discuss the QM equivalence between $f(X)$ and $F_i$ for $i \in \{1,2,3,4,5,6,7,8,12,13,14 \}$. Next, we suppose that $f(X)$ is QM equivalent to $F_9(X)=X^3 + aX^{q+2} + bX^{2q+1} + cX^{3q}$, where $q = 5^m$, $c = 4$, $b = a+2$, $a \neq -1$, and $m$ is odd. Then there exist $u, v \in \mathbb{F}_{q^2}^{*}$ and a positive integer $1 \leq d \leq q^2 - 2$ with $\gcd(d, q^2 - 1) = 1$ such that
$uf(vX^d) = F_9(X)$
or equivalently,
\begin{align*}
 uv^{3q}D_3X^{3qd} + uv^{2q+1}D_2X^{(2q+1)d} + uv^{q+2}D_1X^{(q+2)d}&  + uv^3 D_0X^{3d} = \\
& X^3 + aX^{q+2} + bX^{2q+1} + cX^{3q}.\\
\end{align*}
This implies that the following sets of coefficients of $uf(vX^d)$ and $F_9(X)$ are equal
\[
    A := \{uv^{3q}D_3, uv^{2q+1}D_2, uv^{q+2}D_1, uv^3 D_0\} = \{1, a, b, c\} .
\]
It follows that one of the coefficients from set $A$ must be $1$ and one must be $c=4$. First assume that $uv^3 D_0 = 1$ and $uv^{3q}D_3 = c = 4$, which gives
\[
    \frac{D_0}{v^{3(q-1)}D_3} = -1.
\]
Raising both sides to the $(q+1)$-th power and substituting the explicit values of $D_0$ and $D_3$, we get
\[
    \frac{{D_0}^{q+1}}{{D_3}^{q+1}} = \frac{(\delta^3 - \tilde{\delta}^q)(\delta^{3q} - \tilde{\delta})}{(\delta^{3q} - \tilde{\delta}^q)(\delta^3 - \tilde{\delta})}  = 1.
\]
This simplifies to $(\delta^3 - \delta^{3q})(\tilde{\delta}^q - \tilde{\delta}) = 0$. Since $\tilde{\delta} \notin \mathbb{F}_q$, we deduce that $\delta^3 = \delta^{3q}$, which is impossible for all $\delta \in \mathbb{F}_{q^2} \setminus \mathbb{F}_q$ because one can always select $\delta \in \mathbb{F}_{q^2} \setminus \mathbb{F}_q$ such that $\delta^3 \notin \mathbb{F}_q$ for large $q$. Now, let us assume that $uv^3 D_0=1$ and $uv^{q+2}D_1=c=4$ which implies that $\frac{{D_0}^{q+1}}{{D_1}^{q+1}}=1$. Then substituting the values of $D_0$, $D_1$, we get
 \[
 \frac{{D_0}^{q+1}}{{D_1}^{q+1}}=\dfrac{-(\delta^{3}-{\tilde \delta}^q)(\delta^{3q}-{\tilde \delta})}{(\delta^{2q+1}-{\tilde \delta})(\delta^{q+2}-{\tilde \delta}^q)}=1.
 \]
This further leads to the following equation
 \[
 (\delta^{3}-{\tilde \delta}^q)(\delta^{3q}-{\tilde \delta})+(\delta^{2q+1}-{\tilde \delta})(\delta^{q+2}-{\tilde \delta}^q)=0.
 \]
Substituting  ${\tilde \delta}=2\delta^3$, we obtain
 \(
 \delta^6-\delta^{q+5}+\delta^{6q}-\delta^{5q+1}=0,
 \)
 which implies $(\delta-\delta^q)(\delta^5-\delta^{5q})=0$. This can only hold if $\delta^5 \in \F_q$ or equivalently $\delta \in \F_q$, which is not possible. One can use the similar techniques as above for the remaining cases. Similarly, we can show $f(X)$ is QM inequivalent to the permutation quadrinomial $F_{10}(X)$. 
 
 Next, we show that $f(X)$ is QM inequivalent to $F_{11}(X)=X^3+aX^{q+2}+bX^{2q+1}+cX^{3q},$ where $q=5^m, c=a+2,b=2a, a+2$ is a square element of $\F_{q}$ and $m$ is odd. On the contrary, assume that $f(X)$ is QM equivalent to $X^3+aX^{q+2}+bX^{2q+1}+cX^{3q}.$ Then, we have $\{uv^{2q+1}D_2,uv^{3q}D_3,uv^3 D_0,uv^{q+2}D_1\}=\{c,b,a,1\}$. First, suppose that $(uv^{2q+1}D_2,uv^{3q}D_3,uv^3 D_0,uv^{q+2}D_1)=(c,b,a,1)$. Therefore, by using the condition $b=2a$, we get $v^{3q}D_3=2v^{3}D_0$, or equivalently, $2\left(\frac{D_0}{D_3}\right)=v^{3q-3}$. Raising the previous expression by $q+1$, we obtain 
 \[
 \frac{{D_0}^{q+1}}{{D_3}^{q+1}}=\dfrac{(\delta^{3}-{\tilde \delta}^q)(\delta^{3q}-{\tilde \delta})}{(\delta^{3q}-{\tilde \delta}^q)(\delta^3-{\tilde \delta})}=-1.
 \]
 We now choose $\tilde \delta=2\delta^3$ to get $(\delta^3+\delta^{3q})^2=0$, that is, $\delta^{3q}=-\delta^3$ or,  $\delta^{3q-3}=-1$. Squaring the last expression, we get $\delta^{6q-6}=1$, which does not hold for all $\delta \in \F_{q^2} \setminus \F_{q}$ for sufficiently large $q$.  We can apply the similar technique for the remaining possibilities.
 
We now show that $f(X)$ is QM inequivalent to $F_{15}(X)=X^3+aX^{q+2}+bX^{2q+1}+cX^{3q}$, where $p=5, c=(-1)^t+1,b=(-1)^ta+2$ and $a^{p^m}+a+2(-1)^t \neq 0$. If $t$ is odd, then $c=0$. Thus, there is no need to discuss the QM equivalence of $f(X)$ with $F_{15}(X)$ in this case. Let us now assume $t$ is even, then $c=2$. In this case, one can use the argument used for showing the QM inequivalence of permutation quadrinomial $f(X)$  with $F_{11}(X)$ and see that $f(X)$ is QM inequivalent to $F_{15}(X)$. For the permutation quadrinomial $F_{16}(X)=a_1X+a_2X^{s_1(p^m-1)+1}+X^{s_2(p^m-1)+1}+a_3X^{s_3(p^m-1)+1}$, we use SageMath to verify that it is QM inequivalent to $f(X)$.

We will now discuss the QM equivalence of $f(X)$ with permutation quadrinomial $F_{17}(X)=a_1X+a_2X^{s_1(p^m-1)+1}+X^{s_2(p^m-1)+1}+a_3X^{s_3(p^m-1)+1}$, where   $(s_1,s_2,s_3)=(\frac{-1}{p^k-2},1,\frac{p^k-1}{p^k-2})$, $a_1 \notin \mu_{q+1}, a_2^{p^m}=\frac{a_3}{a_1}\in \mu_{q+1}$, and $\left(-\frac{a_3}{a_1}\right)^{\frac{p^{m}+1}{\gcd(p^k-1,p^m+1)}}\neq 1$.
 Suppose that $f(X)$ is QM equivalent to $F_{17}(X)$. Then
 \[
    A' := \{uv^{3q}D_3, uv^{2q+1}D_2, uv^{q+2}D_1, uv^3 D_0\} = \{1, a_1, a_2, a_3\} .
\]
It is sufficient to show that for any two choices of $a_3$ and $a_1 \in A'$, we can always choose $\delta, \tilde \delta \in \F_{q^2} \setminus \F_q$ such that $\frac{a_3}{a_1} \not \in \mu_{q+1}$. Here, we only show for $a_3=uv^{3q}D_3$ and $a_1=uv^{2q+1}D_2$, and the remaining cases can be done in a similar manner. On the contrary, assume that 
\[
\left(\dfrac{a_3}{a_1}\right)^{q+1}=\left(\dfrac{\delta^3-\tilde\delta}{(-3)(\delta^{q+2}-\tilde\delta)}\right)^{q+1}=1.
\]
 This give us $(\delta^{3q}-\tilde \delta^q)(\delta^3-\tilde\delta)=9(\delta^{2q+1}-\tilde \delta^q)(\delta^{q+2}-\tilde\delta)$. One can take $\tilde \delta=-\delta^3$, to get $4\delta^{3q+3}=9\delta^{2q+2}(\delta+\delta^q)^2$, which implies $9\delta^{2q}+9\delta^{2}+5\delta^{q+1}=0$.  In particular, when $p=5$, we obtain that $\delta^{2q}+\delta^2=0$. However, it is not always true as one can choose $\delta \in \F_{q^2}\setminus \F_{q}$ such that $\delta^{4q-4}\neq 1$.  This methodology also cover the study of QM equivalence of $f(X)$ with the permutation quadrinomial $F_{18}(X)$.
 
Now, we show that $f(X)$ is QM inequivalent to the permutation quadrinomial $F_{19}(X)$. On the contrary, suppose $f(X)$ and $F_{19}(X)=cX^{3q}+bX^{2q+1}+aX^{q+2}+X^3 \in \F_q[X]$ are QM equivalent. Then there exist $u,v \in \F_{q^2}^{*}$ and a positive integer $1 \leq d < q^2-1$ with $\gcd(d,q^2-1)=1$ such that $uf(vX^d)=cX^{3q}+bX^{2q+1}+aX^{q+2}+X^3.$ Therefore, the coefficients of the polynomial $uf(vX^d)$ are from the set $\{1,a,b,c\}$. Let $uv^{3q}D_3=1$, then we get that the remaining coefficients of $uf(vX^d)$ as $\frac{D_2}{D_3} v^{-q+1},  \frac{D_1}{D_3} v^{-2q+2}$ and $\frac{D_0}{D_3} v^{-3q+3}$ are in $\F_q$. Hence, $$\left(\frac{D_2}{D_3} v^{-q+1}\right) \left(\frac{D_1}{D_3} v^{-2q+2}\right)  \left(\frac{D_0}{D_3} v^{-3q+3}\right)^{-1} \in \F_q.$$
 Thus, we have $\frac{D_1 D_2}{D_0 D_3} \in \F_q$. Substituting the values of $D_0, D_1, D_2$ and $D_3$, we obtain 
 $$\dfrac{(\delta^{2q+1}-\tilde\delta)(\delta^{q+2}-\tilde\delta)}{(\delta^{3q}-\tilde\delta)(\delta^{3}-\tilde\delta)} \in \F_q.$$
Using SageMath, we get several choices for $\delta, \tilde \delta \in \F_{q^2} \setminus \F_q$ such that the above condition is not satisfied.
This completes the proof.
\end{proof}

\begin{rmk}
 In a similar way to the Theorem~\ref{quadeq1}, one can see that the permutation quadrinomials obtained in Theorem~\ref{T15} are  QM inequivalent to the known permutation quadrinomials in Table~\ref{Table0}
\end{rmk}

\begin{thm}\label{quadeq2}
 Let $f(X)$ be as in Theorem~\ref{T50} then $f(X)$ is QM inequivalent to the all permutation quadrinomials $F_i$ listed in Table~\ref{Table0} for $1 \leq i \leq 19$ over $\F_{q^2}$.
\end{thm}
\begin{proof}
 We first consider the permutation quadrinomial $F_1(X)$. Assume that \( f(X) \) is QM equivalent to  
\[
F_1(X)=X^3 + aX^{q+2} + bX^{2q+1} + cX^{3q},
\]  
when \( p = 3 \), \( b = -a \), \( c = a \neq -1 \), and \( a^{\frac{q-1}{2}} = 1 \). Hence, there exist \( u, v \in \mathbb{F}_{q^2}^{*} \) and a positive integer \( 1 \leq d \leq q^2-2 \) with \( \gcd(d, q^2-1) = 1 \) such that  
\[
uf(vX^d) = X^3 + aX^{q+2} + bX^{2q+1} + cX^{3q}.
\]
This further implies that  
\[
A'' := \{uv^{3q}D_3, uv^{2q+1}D_2, uv^{q+2}D_1, uv^3 D_0\} = \{1, a, b, c\}.
\]
First, suppose that \( uD_0 v^3 = 1 \), \( a = uD_3 v^{3q} \), and \( b = uD_1 v^{q+2} \). Computing \( b = -a \), we get \( \frac{D_1}{D_3} = -v^{2q-2} \). Raising \( \frac{D_1}{D_3} = -v^{2q-2} \) to the power \( q+1 \), we obtain  
\[
\left(\frac{D_1}{D_3}\right)^{q+1} = \left(\frac{\alpha \beta^2 (\delta^q - \delta)}{\delta^3 - \alpha\delta - \tilde{\delta}}\right)^{q+1} = 1.
\]
One can verify that for \( \alpha = -1 \), \( \tilde{\delta} = \delta^3 \), and \( m \) odd, the above expression renders \( \delta^{2q-2} = 1 \), which is not true for all \( \delta \in \mathbb{F}_{q^2} \setminus \mathbb{F}_{q} \) when \( q \) is large. Using a similar argument, we can show that \( f(X) \) is QM inequivalent to  
\(
X^3 + aX^{q+2} + bX^{2q+1} + cX^{3q}
\)
for other possible choices of \( a, b, c \) from the set \( A'' \).  

Notice that the permutation quadrinomials $F_i$ for $i \in \{2,3,4,5,6,7,8,12,13,14 \}$ are QM inequivalent to \( f(X) \) by applying the similar technique as above, since \( b = \pm a \). Next, consider the permutation quadrinomial  
\[
F_{16}(X)=a_1X + a_2X^{s_1(p^m-1)+1} + X^{s_2(p^m-1)+1} + a_3X^{s_3(p^m-1)+1},
\]
where \( p \) is odd, \( \gcd(3, p-1) = 1 \), and  
\(
\theta_1(2\theta_4 + \theta_3 - 3\theta_1) = \theta_4(\theta_3 - \theta_4),
\)
with \( \theta_1 \in \mathbb{F}_{p^m}^{*} \), \( \theta_2 \in \mathbb{F}_{p^m} \), and \( \theta_2^2 - 4\theta_1\theta_4 \) being a square in \( \mathbb{F}_{p^m}^{*} \), where  
\(
\theta_1 = a_1a_3^{p^m} - a_2, \quad \theta_2 = a_1a_2^{p^m} - a_3, \quad \theta_3 = a_1^{p^m+1} + a_2^{p^m+1} - a_3^{p^m+1} - 1, \quad \theta_4 = a_1^{p^m+1} - 1.
\)
It is confirmed by experiments that there exist some \( \delta \) and \( \tilde{\delta} \) in \( \mathbb{F}_{q^2} \setminus \mathbb{F}_q \) for which this polynomial is QM inequivalent to \( f(X) \).  

The argument used to show the QM inequivalence of \( f(X) \) with the permutation quadrinomials $F_i$ for $17 \leq i \leq 19$ follows along similar lines as in Theorem~\ref{quadeq1}.
\end{proof}

\begin{rmk}
 Notice that the permutation quadrinomials obtained in Theorem~\ref{T51} are  QM inequivalent to the known permutation quadrinomials in Table~\ref{Table0} using the similar arguments as used in Theorem~\ref{quadeq2}.
\end{rmk}

We now discuss the QM equivalence of proposed classes of permutation pentanomials in Section~\ref{S4} with the known classes of permutation pentanomials listed in Table~\ref{Table1} in Appendix~\ref{pent}.
\begin{thm}\label{penteq1}
 Let $f(X)$ be as in Theorem~\ref{T60}. Then $f(X)$ is QM inequivalent to the all permutation pentanomials $G_i$ listed in Table~\ref{Table1} for $1 \leq i \leq 36$ over $\F_{q^2}$. 
\end{thm}
\begin{proof}
We will first see that $f(X)$ is QM inequivalent to the permutation pentanomials $G_i$ for $1 \leq i \leq 17$, $19 \leq i \leq 31$ and $i=33$ of Table \ref{Table1}. For this, we will only show the QM inequivalence of $f(X)$ with the permutation pentanomial $G_1(X)$ . For the permutation pentanomials $G_i$ for $2 \leq i \leq 17$, $19 \leq i \leq 31$ and $i=33$, a similar approach can be followed.
Let us assume that $f(X)$ is QM equivalent to $G_1(X)=X^5 + X^{2^m+4} + X^{3 \cdot 2^m+2} + X^{4 \cdot 2^m +1} + X^{5 \cdot 2^m}$ when $m \not \equiv 0 \pmod 4$. Hence, there exist $u,v\in \F_{q^2}^{*}$ and a positive integer $1 \leq d \leq q^2-2$ with $\gcd(d,q^2-1)=1$ such that 
 $$uf(vX^d)=X^5 + X^{2^m+4} + X^{3 \cdot 2^m+2} + X^{4 \cdot 2^m +1} + X^{5 \cdot 2^m}.$$
On comparing the coefficients in the above equation, we obtain the following equality $$uv^{4q}(N_4+\tilde \delta)=uv^{3q+1}N_3=uv^{2q+2}N_2=uv^{q+3} N_1=uv^4(N_0+\beta^4\tilde\delta)=1,$$
where $N_0, N_1, N_2, N_3$ and $N_4$ are coming from Theorem~\ref{T60}. It is easy to observe that if $uv^{4q}(N_4+\tilde \delta)=uv^{3q+1}N_3$ then $\frac{N_4+\tilde \delta}{N_3}=v^{1-q}$. Raising $\frac{N_4+\tilde \delta}{N_3}=v^{1-q}$ to the power $q+1$, we obtain $\left(\frac{N_4+\tilde \delta}{N_3}\right)^{q+1}=1$. Substituting the values of $N_3=\alpha\beta(\delta+\delta^q)$ and $N_4=\delta^4+\delta^2+\alpha\delta$ in the expression $\left(\frac{N_4+\tilde \delta}{N_3}\right)^{q+1}=1$, we get $(\delta^{4q}+\delta^{2q})(\delta^{4}+\delta^{2})=\delta^2+\delta^{2q}$ by setting $\tilde \delta=\delta$ and $\alpha=1$. This further implies that $(\delta^{2q+2}+1)(\delta^{2q+2}+\delta^{2q}+\delta^2)=0$. Thus, we have either $\delta^{q+1}=1$ or $\delta^{q+1}+\delta^q+\delta=0$. Therefore, $\delta$ can have at most a total of $2(q+1)$ choices in $\F_{q^2}$. However, we have $q^2-q > 2(q+1)$ choices of $\delta \in \F_{q^2} \setminus \F_q$ for $q \geq 4$. This shows that $f$ is QM inequivalent to the permutation pentanomial $G_1(X)$.

We are now left with determining the QM equivalence between $f(X)$ and the permutation pentanomials $G_i(X)$ for $i \in \{18, 32, 34, 35, 36\}$ listed in Table \ref{Table1}. Denote the set of exponents of $uf(vX^d)$ by $C:=\{4d,(q+3)d,(2q+2)d,(3q+1)d,4qd\}$. It is evident that $G_{18}(X)$ is a permutation polynomial if and only if $G_{18}(X^4)$ is a permutation polynomial. Note that the set of exponents of $G_{18}(X^4)$ is $D:=\{4,q+3,2q+2,3q+1,4q\}$. On comparing the sets $C$ and $D$, after some computations, we obtained that $d$ is either $1$ or $q$. Furthermore for $d=1$ and $d=q$, by using SageMath, we  get that there exist several $\delta, \tilde \delta \in\F_{q^2} \setminus \F_q$ for which $uf(vX^d)$ does not satisfy all the conditions required for $G_{18}(X)$ to be a permutation pentanomial. Thus, $f(X)$ is QM inequivalent to the permutation pentanomial $G_{18}(X)$. We now verify  the QM inequivalence of $f(X)$ with the permutation pentanomial $G_{32}(X)$ of Table \ref{Table1}. It is easy to see that some of the coefficients of the permutation pentanomial $G_{32}(X)$ are equal. Hence, one can use the similar approach as used for the permutation pentanomials $G_{1}(X)$ to show that $f(X)$ is QM inequivalent to the permutation pentanomial $G_{32}(X)$. 

Next, we consider the permutation pentanomial $G_{34}(X)$ listed in Table~\ref{Table1}. Recall that \( G_{34}(X) = X^{4q} + aX^{3q+1} + bX^{2q+2} + cX^{q+3} + dX^{4} \in \mathbb{F}_q[X] \) permutes \( \mathbb{F}_{q^2} \) if and only if either:
\begin{enumerate}
    \item \( a = c \), \( d \neq 1 \), and the polynomial \( X^3 + \frac{a+b}{b+d+1}X + \frac{a}{b+d+1} \) has no root in \( \mathbb{F}_q^* \); or  
    \item \( a \neq c \), \( b+d+1 = c+ad = 0 \), and \( \operatorname{Tr}_1^m \left(\frac{b}{a+c}\right) = 0 \).
\end{enumerate}
Now, suppose \( f(X) \) is QM equivalent to \(G_{34}(X) \). Then, there exist \( u, v \in \mathbb{F}_{q^2}^* \) and an integer \( 1 \leq d \leq q^2 - 2 \) with \( \gcd(d, q^2 - 1) = 1 \) such that  
\[
uf(vX^d) = X^{4q} + aX^{3q+1} + bX^{2q+2} + cX^{q+3} + dX^4.
\]  
This implies that the following sets of coefficients of \( uf(vX^d) \) and \( G_{34}(X) \) must be equal
\[
B := \{ uv^{4q} (N_4 + \tilde{\delta}), uv^{3q+1} N_3, uv^{2q+2} N_2, uv^{q+3} N_1, uv^4 (N_0 + \beta^4 \tilde{\delta}) \} = \{1, a, b, c, d\}.
\]  
If \( a = c \), then two coefficients in the set \( B \) are equal. In this case, one can proceed similarly to the analysis of the permutation pentanomial $G_{1}(X)$.  

Now, assume \( a \neq c \). Using the condition \( c + ad = 0 \), we show that \( f(X) \) and \(G_{34}(X) \) are QM inequivalent. Setting  
\[
uv^{4q} (N_4 + \tilde{\delta}) = 1, \quad a = uv^{2q+2} N_2, \quad c = uv^{3q+1} N_3, \quad d = uv^4 (N_0 + \beta^4 \tilde{\delta}),
\]  
the equation \( c + ad = 0 \) simplifies to  
\[
\left( \frac{N_3}{N_4 + \tilde{\delta}} \right) v^{-q+1} = \left( \frac{N_2}{N_4 + \tilde{\delta}} \right) v^{-2q+2} \left( \frac{N_0 + \beta^4 \tilde{\delta}}{N_4 + \tilde{\delta}} \right) v^{-4q+4}.
\]  
Raising both sides to the power \( q+1 \) and simplifying, we obtain  
\[
\left( \frac{N_2 (N_0 + \beta^4 \tilde{\delta})}{N_3 (N_4 + \tilde{\delta})} \right)^{q+1} = 1.
\]  
Substituting the explicit values of \( N_0, N_2, N_3, \) and \( N_4 \), we derive  
\[
\left( \frac{(\delta + \delta^q + \alpha)(\delta^{4q} + \delta^{2q} + \alpha \delta^q + \beta^4 \tilde{\delta})}{\alpha(\delta^4 + \delta^2 + \alpha \delta + \tilde{\delta})} \right)^{q+1} = 1.
\]  
By setting \( \tilde{\delta} = \delta^4 \), we find that \( \delta \) satisfies a polynomial of degree \( 10q \) over \( \mathbb{F}_{q^2} \). Hence, for \( q > 11 \), we can always choose \( \delta \in \mathbb{F}_{q^2} \setminus \mathbb{F}_q \) such that \( c + ad \neq 0 \). A similar argument applies to all other possible values of \( a, b, c, \) and \( d \).  

For the permutation pentanomials $G_{35}(X)$ and $G_{36}(X)$, a similar methodology as used for the permutation pentanomial $G_{1}(X)$ can be applied to show that these pentanomials are not QM equivalent to \( f(X) \). This completes the proof of the desired result.
\end{proof}

\begin{rmk}
  The permutation pentanomials obtained in Theorem~\ref{T61} and Theorem~\ref{T64} are QM inequivalent to the known permutation pentanomials in Table \ref{Table1} using the similar arguments as used in Theorem~\ref{penteq1}.
\end{rmk}

\begin{thm}\label{penteq2}
 The permutation pentanomials obtained in Theorem~\ref{T65}, Theorem~\ref{T66} and Theorem~\ref{T69} are QM inequivalent to the all permutation pentanomials listed in Table \ref{Table1}.
\end{thm}
\begin{proof}
 The proof follows along the similar lines as of Theorem~\ref{penteq1}.
\end{proof}

Finally, we show that the obtained classes of permutation binomials in Theorem \ref{T70} are new by studying their QM equivalence with the known permutation binomials listed in Table \ref{Table2} in Appendix~\ref{bino}.

\begin{thm}\label{binomialq1}
 The permutation binomials obtained in Theorem~\ref{T70} are QM inequivalent to the all permutation binomials $H_i$ listed in Table~\ref{Table2} for $1 \leq i \leq 8$ over $\F_{q^2}$.
\end{thm}
\begin{proof}
Here, we discuss the QM equivalence only for the first family of permutation binomials introduced in Theorem~\ref{T70} with the permutation binomials listed in Table \ref{Table2}. One can use the similar methodology to show that the second family of permutation binomials in Theorem~\ref{T70} is QM inequivalent to permutation binomials listed in Table \ref{Table2}. Recall that, $f(X)=(\beta^4 \delta^{4q} +\beta^4  {\tilde \delta}) X^{3-q}+(\delta^4+{\tilde \delta})X^{3q-1}$ is the first family of permutation binomials we obtained.  Suppose that $f(X)$ is QM equivalent to the permutation binomial $H_1(X)$ mentioned in Table \ref{Table2}, then there must exist a positive integer $1 \leq d \leq q^2-2$ with $\gcd(d,q^2-1)=1$ and $u, v \in \F_{q^2}^{*}$ such that $f(X)= uH_1(vX^d)$. This further implies that the set of exponents of $f(X)$ and $uH_1(vX^d)$ are equal in $\mathbb{Z}_{q^2-1}$, that is, either $d \equiv 3-q \pmod {q^2-1}$ or
$d \equiv 3q-1 \pmod {q^2-1}.$ WLOG, we assume $d \equiv 3-q \pmod {q^2-1}$ and thus $d(q+2) \equiv 3q-1 \pmod {q^2-1}$. This leads to the equation $(q+2)(3-q) \equiv 3q-1 \pmod {q^2-1}$. We further obtain $2(q-1) \equiv 4  \pmod {q^2-1}$, a contradiction to the fact that $q-1 \nmid 4$ for $q>2$. Hence, $f(X)$ is QM inequivalent to the permutation binomial $H_1(X)$. Using the similar approach, one can see that the permutation binomials
$H_i(X)$ for $i \in \{3,4,5,6,8\}$ listed in Table \ref{Table2} are QM inequivalent to $f(X)$. 

We next consider the permutation binomial  $H_2(X):=X^{\frac{2^n-1}{2^t-1}+1}+ aX$, $n=2^st$,$s \in \{1,2\}$, $t$ is odd and $a\in \omega \F_{2^{t}}^{*} \cup \omega^2 \F_{2^{t}}^{*}$, where $\omega$ is primitive third root of unity. Notice that for $s=1$, the sets of exponents of permutation binomials $H_1(X)$ and $H_2(X)$ are exactly same. Hence, we are done for $s=1$. Let $s=2$ and $f(X)$ is QM equivalent to $H_2(X)$. Therefore, the set of exponents $\{3-q,3q-1\}$ and $\{d(\frac{2^n-1}{2^t-1}+1), d\}$ of $f(X)$ and $uH_2(vX^d)$, respectively, are same in $\mathbb{Z}_{q^2-1}$, where $1 \leq d \leq q^2-2$ is an integer such that $\gcd(d,q^2-1)=1$. If $d \equiv 3-q \pmod {q^2-1}$, then we get the following expression
$$((2^t+1)(2^{2t}+1)+1)(3-q) \equiv 3q-1 \pmod {q^2-1},$$
or equivalently,
$$(2^t+1)(q+1)(3-q) \equiv 4q-4 \pmod {q^2-1}.$$
Since $q-1$ is relatively prime to $q+1$ and $3-q$, we obtain that $(q-1) \mid 2^t+1= q^{1/2}+1,$ which is not possible for sufficiently large $q$. A similar argument can be applied in the case when $d \equiv 3q-1 \pmod {q^2-1}$. Finally, we show that $f(X)$ is QM inequivalent to the permutation binomial $H_7(X)=X^{d'}+aX$, $n=rk$, $d'=\frac{2^{rk}-1}{2^k-1}, \gcd(d'-1,2^k-1)=\gcd(r,2^k-1)=1$ and $a \not \in \F_{2^k}^{*}$. On the contrary, let $f(X)$ and $H_7(X)$ be QM equivalent, which leads to existence of an interger $1 \leq d \leq q^2-2$ with $\gcd(d,q^2-1)=1$ such that  $\{3-q,3q-1\}=\{dd', d\}$ in $\mathbb{Z}_{q^2-1}$. WLOG, suppose that $d \equiv 3-q \pmod {q^2-1}$ and $dd'\equiv 3q-1 \pmod {q^2-1}$, which renders 
$$\frac{2^{rk}-1}{2^k-1}(3-q) \equiv 4(q-1) \pmod {q^2-1},$$
that is,
$$\frac{q^2-1}{2^k-1}(3-q) \equiv 4(q-1) \pmod {q^2-1}.$$
This implies $\frac{q^2-1}{2^k-1}(3-q) \equiv 0 \pmod {q-1}$ and we further obtain $(q-1) \mid \frac{q-1}{2^k-1}$, a contradiction.
\end{proof}

\section{Conclusion}\label{S6}
In this paper, we have constructed several classes of permutation polynomials with a few terms (two classes of binomials, six classes of quadrinomials and six classes of pentanomials) of nontrivial coefficients over the finite field $\F_{q^2}$. Our main ingredient was the well classified permutation rational functions of degree at most $4$, which permute the projective line $\bP_1(\F_q)$. Using these permutation rational functions, we have constructed bijections on the unit circle $\mu_{q+1}$ to produce new classes of permutation polynomials with a few terms over $\F_{q^2}$. We also discuss the QM-equivalence of the proposed classes of permutation polynomials with the known ones in the literature. For future work, one may be interested in constructing new classes of permutation polynomials using the remaining cases of the permutation rational functions of degree $4$, and by using the recent classification~\cite{CSze} of degree $5$ permutation rational functions.

\section*{Acknowledgments}
We sincerely thank the editors for handling our paper, and the referees for their careful reading, valuable comments, and constructive suggestions.

\appendix
\section{permutation quadrinomials}\label{quad}
 \begin{longtable}{|m{2em}|m{17em}|m{11em}|m{4em}|}
 \caption{Known classes of permutation quadrinomials over $\F_{q^2}$ where $q=p^m$ for an odd prime $p$ and the polynomials $F_i$ for $i=1,2,\dots, 16$ have the same form.\label{Table0}}
 \\
 \hline
$F_i$& Polynomials & Conditions & Reference\\
\hline
$F_1$& $X^3 + aX^{q+2}+bX^{2q+1}+cX^{3q}$  & $p = 3, b = -a, c = a \neq -1$ and $a^{\frac{p^m-1}{2}}=1$ & ~\cite{Gupta20}\\
\hline
 $F_2$&   & $p = 3, b = -a, c = a -1$ and $(-a)^{\frac{p^m-1}{2}}=1$ & ~\cite{Gupta20}\\
\hline
 $F_3$&   & $p = 3, b = -a, c = 1-a$, $a \neq -1$ and $m$ is even & ~\cite{Gupta20}\\
 \hline
  $F_4$&   & $p = 3, b = -a, c = 1$ & ~\cite{Gupta20}\\
  \hline
  $F_5$&    & $p = 3, b = a, c = a \neq 1$ and $(-a)^{\frac{p^m-1}{2}}=1$ & ~\cite{Gupta20}\\
\hline
$F_6$&    & $p = 3, b = a, c = a+1$ and $a^{\frac{p^m-1}{2}}=1$ & ~\cite{Gupta20}\\
\hline
$F_7$&    & $p = 3, b = a, c = 2a+1, a\neq1$ and $m$ is even  & ~\cite{Gupta20}\\
\hline
$F_8$&   & $p = 3, b = a, c =- 1$ & ~\cite{Gupta20}\\
\hline
$F_9$&    & $p = 5, c=4, b = a+2, a \neq -1$ and $m$ is odd & ~\cite{Gupta20}\\
\hline
$F_{10}$&   & $p = 5, c=1, b = 2-a, a \neq 1$ and $m$ is odd & ~\cite{Gupta20}\\
\hline
$F_{11}$&   & $p = 5, c=a+2, b = 2a$, $a+2$ is a square of $\F_{5^m}$ and $m$ is odd & ~\cite{Gupta20}\\
\hline
$F_{12}$&    & $p = 3, b = (-1)^t a, (c - (-1)^t)(c-a + (-1)^t) \neq 0$ and $\epsilon + \epsilon^{p^m} \neq 0$, where $\epsilon$ is a square root of $\frac{2a}{c-a+(-1)^t}$ & ~\cite{Gupta22}\\
\hline
$F_{13}$&    & $p = 3, b = (-1)^t a, c=(-1)^{t+1}$ and $a+a^{p^m} \neq 0$ & ~\cite{Gupta22}\\
\hline
$F_{14}$&    & $p = 3, b = (-1)^t a, c=-((-1)^{t}+a)$, $a+a^{p^m}+(-1)^{t}a^{p^m+1} \neq 0$ and $m$ is even & ~\cite{Gupta22}\\
\hline
$F_{15}$&   & $p = 5, b = (-1)^t a+2, c=(-1)^{t}+1$, $a+a^{p^m}+2(-1)^{t} \neq 0$ & ~\cite{Gupta22}\\
\hline 
$F_{16}$&    & $p$ is odd, $\gcd(3,p-1)=1$,$\theta_1(2\theta_4+\theta_3-3\theta_1)=\theta_4(\theta_3-\theta_4),\theta_1 \in \F_{p^m}^{*},\theta_2 \in \F_{p^m}$ and $\theta_2^2-4\theta_1\theta_4$ is a square in $\F_{p^m}^{*},$ where $\theta_1=a_1a_3^{p^m}-a_2,\theta_2= a_1a_2^{p^m}-a_3,\theta_3=a_1^{p^m+1}+a_2^{p^m+1}-a_3^{p^m+1}-1,\theta_4=a_1^{p^m+1}-1$ & ~\cite{Chen}\\
\hline 

$F_{17}$& $a_1X+a_2X^{s_1(q-1)+1}+X^{s_2(q-1)+1}+a_3X^{s_3(q-1)+1}$, where $(s_1,s_2,s_3)=( \frac{-1}{p^k-2},1,\frac{p^k-1}{p^k-2})$& $a_1 \notin U, a_2^{p^m}=\frac{a_3}{a_1}\in U$, and $\left(-\frac{a_3}{a_1}\right)^{\frac{p^{m}+1}{\gcd(p^k-1,p^m+1)}}\neq 1$   & ~\cite{Chen}\\
\hline
$F_{18}$& $a_1X+a_2X^{s_1(p^m-1)+1}+X^{s_2(p^m-1)+1}+a_3X^{s_3(p^m-1)+1}$,   where $(s_1,s_2,s_3)=( \frac{p^k+1}{p^k+2},1,\frac{1}{p^k+2})$& $a_1 \notin U, a_3^{p^m}=\frac{a_2}{a_1}\in U$, and $\left(-\frac{a_2}{a_1}\right)^{\frac{p^{m}+1}{\gcd(p^k+1,p^m+1)}}\neq 1$   & ~\cite{Chen}\\
 \hline
$F_{19}$&   $X^3 + aX^{q+2}+bX^{2q+1}+cX^{3q}$ & see Theorems 2 and 3   & ~\cite{FB_quad_2023}\\
 \hline
\end{longtable}

\section{permutation pentanomials}\label{pent}
\begin{longtable}{|m{2em}|m{21em}|m{9em}|m{2em}|}
\caption{Known permutation pentanomials over $\F_{2^{2m}}$ \label{Table1}}\\
 \hline
$G_{i}$ & Polynomials & Conditions & Ref.\\
\hline
$G_{1}$& $X^5 + X^{2^m+4} + X^{3 \cdot 2^m+2} + X^{4 \cdot 2^m +1} + X^{5 \cdot 2^m}$ & $m \not \equiv 0 \pmod 4$ & \cite{Xu_penta_18} \\
\hline
$G_{2}$& $X^3 + X^{2^{m+1}+1} + X^{3 \cdot 2^m} + X^{4 \cdot 2^m -1} + X^{-2^m+4}$ & $m$ is odd & \cite{Xu_penta_18}\\
\hline

$G_{3}$& $X^5 + X^{2^m+4} + X^{2 \cdot 2^m+3} + X^{4 \cdot 2^m +1} + X^{5 \cdot 2^m}$ & $m \not \equiv 0 \pmod 4$ & \cite{Xu_penta_18}\\
\hline
$G_{4}$& $X^3 + X^{2^m+2} + X^{3 \cdot 2^m} + X^{4 \cdot 2^m -1} + X^{-2^m+4}$ & $m$ is odd &\cite{Xu_penta_18} \\
\hline
$G_{5}$& $X^7 + X^{2 \cdot 2^m+5} + X^{3 \cdot 2^m+4} + X^{5 \cdot 2^m +2} + X^{6 \cdot 2^m+1}$ & $\gcd(m,3)=1$ & \cite{Xu_penta_18}\\
\hline
$G_{6}$& $X^5 + X^{2^m+4} + X^{3 \cdot 2^m+2} + X^{4 \cdot 2^m +1} + X^{6 \cdot 2^m-1}$ & see Theorem 3.6 & \cite{Xu_penta_18}\\
\hline
$G_{7}$& $X^5 + X^{3 \cdot 2^m+2} + X^{4 \cdot 2^m+1} + X^{5 \cdot 2^m} + X^{6 \cdot 2^m-1}$ & $m \equiv 2 \pmod 4$ & \cite{Xu_penta_18}\\
\hline
$G_{8}$& $X^7 + X^{3 \cdot 2^m+4} + X^{4 \cdot 2^m+3} + X^{5 \cdot 2^m+2} + X^{6 \cdot 2^m+1}$ & $m \equiv 0 \pmod 4$ & \cite{Xu_penta_18}\\
\hline
$G_{9}$& $X^9 + X^{3 \cdot 2^m+6} + X^{6 \cdot 2^m+3} + X^{7 \cdot 2^m+2} + X^{9 \cdot 2^m}$ & $m$ is odd & \cite{Xu_penta_18}\\
\hline
$G_{10}$ & $X^9 + X^{2 \cdot 2^m+7} + X^{3 \cdot 2^m+6} + X^{6 \cdot 2^m+3} + X^{9 \cdot 2^m}$ & $m$ is odd & \cite{Xu_penta_18}\\
\hline
$G_{11}$& $X^{2^m+6} + X^{2 \cdot 2^m+5} + X^{5 \cdot 2^m+2} + X^{8 \cdot 2^m-1} + X^{- 2^m+8}$ & see Theorem 4.5 & \cite{Xu_penta_18}\\
\hline
$G_{12}$& $X^{2 \cdot 2^m+5} + X^{5 \cdot 2^m+2} + X^{6 \cdot 2^m+1} + X^{8 \cdot 2^m-1} + X^{- 2^m+8}$ & see Theorem 4.6 &\cite{Xu_penta_18} \\
\hline
$G_{13}$& $X^{3q-2} + X^{2q-1} + X^{q^2-q+1} + X^{q^2-2q+2} + X$ & all $m$  & \cite{LQW}\\
\hline
$G_{14}$& $X^{7q-5} + X^{3q-1} + X^{q^2-q+2} + X^{q^2-5q+6} + X$ & $m \not \equiv 0 \pmod 7$  & \cite{LQW}\\
\hline
$G_{15}$& $X^{7} + X^{3q+4} + X^{4q+3} + X^{6q+1} + X^{7q}$ & see Example 1(1)  & \cite{Deng}\\
\hline
$G_{16}$& $X^{2^k+3} + X^{3q+2^k} + X^{2^kq+3} + X^{(2^k+1)q+2} + X^{(2^k+2)q+1}$ & see Example 1(4) & \cite{Deng}\\
\hline
$G_{17}$& $X^{2^k+3} + X^{3q+2^k} + X^{(2^k+1)q+2} + X^{(2^k+2)q+1} + X^{(2^k+3)q}$ & see Example 1(5) & \cite{Deng}\\
\hline
$G_{18}$& $X+a_1 X^{\frac{1}{4}(q-1)+1}+a_2 X^{\frac{1}{2}(q-1)+1}+a_3 X^{\frac{3}{4}(q-1)+1}+a_4 X^{q}$ & see Theorem 3.1 & \cite{Zhang_penta_23}\\
\hline
$G_{19}$& $X+X^{\frac{1}{17}(q-1)+1}+X^{\frac{8}{17}(q-1)+1}+ X^{\frac{16}{17}(q-1)+1}+X^{\frac{15}{17}(q-1)+1}$ & $m$ is odd and $\gcd(17,q+1)=1$  & \cite{Zhang_penta_23} \\
\hline
$G_{20}$& $X+X^{\frac{11}{13}(q-1)+1}+X^{\frac{9}{13}(q-1)+1}+ X^{\frac{2}{13}(q-1)+1}+X^{q}$ & $\gcd(13,q+1)=\gcd(5,q+1)=1$  & \cite{Zhang_penta_23}\\
\hline
$G_{21}$& $X^{7 \cdot 2^m+1} + X^{5 \cdot 2^m+3} + X^{3 \cdot 2^m+5} + X^{2^m+7} + X^{8}$ & $\gcd(m,7)=1$ & \cite{QLiu_penta_2023}\\
\hline
$G_{22}$& $X^{6 \cdot 2^m} + X^{4 \cdot 2^m+2} + X^{2 \cdot 2^m+4} + X^{-2^m+5} + X^{6}$ & $m$ is odd and $\gcd(m,7)=1$ & \cite{QLiu_penta_2023} \\
\hline
$G_{23}$& $X^{7 \cdot 2^m-1} + X^{6 \cdot 2^m} + X^{4 \cdot 2^m+2} + X^{2\cdot 2^m+4} + X^{6}$ & $m$ is odd and $\gcd(m,7)=1$ & \cite{QLiu_penta_2023}\\
\hline
$G_{24}$& $X^{6 \cdot 2^m-2} + X^{5 \cdot 2^m-1} + X^{3 \cdot 2^m+1} + X^{2^m+3} + X^{-2^m+5}$ & $\gcd(m,7)=1$ & \cite{QLiu_penta_2023}\\
\hline
$G_{25}$& $X^{9 \cdot 2^m+1} + X^{8 \cdot 2^m+2} + X^{6 \cdot 2^m+4} + X^{4\cdot 2^m+6} + X^{10}$ & $m \equiv 2 \pmod 4$ and $\gcd(m,7)=1$ & \cite{QLiu_penta_2023}\\
\hline
$G_{26}$& $X^{8 \cdot 2^m} + X^{7 \cdot 2^m+1} + X^{5 \cdot 2^m+3} + X^{3\cdot 2^m+5} + X^{-2^m+9}$ & $m $ is even and $\gcd(m,7)=1$ & \cite{QLiu_penta_2023}\\
\hline
$G_{27}$& $X^{9 \cdot 2^m-1} + X^{5 \cdot 2^m+3} + X^{3 \cdot 2^m+5} + X^{\cdot 2^m+7} + X^{8}$ & $m $ is even and $\gcd(m,7)=1$ & \cite{QLiu_penta_2023}\\
\hline
$G_{28}$& $X^{8 \cdot 2^m+1} + X^{7 \cdot 2^m+2} + X^{5 \cdot 2^m+4} + X^{3\cdot 2^m+6} + X^{9}$ & $m $ is odd & \cite{QLiu_penta_2023} \\
\hline
$G_{29}$& $X^{7 \cdot 2^m} + X^{6 \cdot 2^m+1} + X^{4 \cdot 2^m+3} + X^{2\cdot 2^m+5} + X^{-2^m+8}$ &  $\gcd(m,3)=1$ and $m \not \equiv 2 \pmod 4$& \cite{QLiu_penta_2023}\\
\hline
$G_{30}$& $X^{8 \cdot 2^m-1} + X^{5 \cdot 2^m+2} + X^{3 \cdot 2^m+2} + X^{\cdot 2^m+6} + X^{7}$ &  $\gcd(m,3)=1$ and $m \not \equiv 2 \pmod 4$& \cite{QLiu_penta_2023}\\
\hline
$G_{31}$& $X^{7 \cdot 2^m-2} + X^{4 \cdot 2^m+1} + X^{2 \cdot 2^m+3} + X^{- 2^m+6} + X^{5}$ &  $m$ is odd &\cite{QLiu_penta_2023} \\
\hline
$G_{32}$& $a^2 X^{i(6q^2-6q)+1}+a^2 X^{i(6q-6)+1}+(a^2+b^2) X^{i(2q^2-2q)+1}+(a^2+b^2) X^{i(2q-2)+1}+c^2 X$ & see Theorem 2 & \cite{Shen} \\
\hline
$G_{33}$& $X^t +X^{r_1(q-1)+t} +X^{r_2(q-1)+t} +X^{r_3(q-1)+t} +X^{r_4(q-1)+t} $ & see Theorems 1-4 & \cite{KX} \\
\hline
$G_{34}$& $X^{4q}+aX^{3q+1} +bX^{2q+2} +cX^{q+3} +dX^{4} $ & see Theorem 3.3 & \cite{Rai} \\
\hline
$G_{35}$& $X^{q^2-q+3}+aX^{4q-1} +X^{3q} +bX^{2q+1} +aX^{3} $ & see Theorem 3.4 & \cite{Rai} \\
\hline
$G_{36}$& $X^{q^2-q+3}+aX^{4q-1} +bX^{3q} +bX^{q+2} +aX^{3} $ & see Theorem 3.5 & \cite{Rai} \\
\hline
\end{longtable}

 \section{permutation binomials}\label{bino}
\begin{longtable}{|m{3em}|m{10em}|m{15em}|m{4em}|}
\caption{Known permutation binomials over $\F_{2^{n}}$ \label{Table2}}\\
 \hline
$H_{i}$ & Polynomials & Conditions & Reference\\
\hline
$H_{1}$& $X^{q+2}+ b X$ & $n=2m, b \in \F_{q^2}\setminus\F_q$, $b^{3(q-1)}=1$ and $m>1$ is odd& \cite{BZ} \\
\hline
$H_{2}$& $X^{\frac{2^n-1}{2^t-1}+1}+ aX$ & $n=2^st$,$s \in \{1,2\}$, $t$ is odd and $a\in \omega \F_{2^{t}}^{*} \cup \omega^2 \F_{2^{t}}^{*}$, where $\omega$ is primitive third root of unity & \cite{BS} \\
\hline
$H_{3}$& $X^{3q-2}+ aX$ & see Theorem 1.1  & \cite{HL} \\
\hline
$H_{4}$& $X^{r(q-1)+1}+ aX$ & $n=2m, r\in\{5,7\}$, see Theorems 1.1-1.2 & \cite{LS} \\
\hline
$H_{5}$& $X^{2q+3}+ aX$ & see Theorem 3.5 & \cite{SG} \\
\hline
$H_{6}$& $X^{2q+4}+ aX^2$ & see Theorem 3.6 & \cite{SG} \\
\hline
$H_{7}$& $X^{d'}+aX$ & $n=rk$, $d'=\frac{2^{rk}-1}{2^k-1}, \gcd(d'-1,2^k-1)=\gcd(r,2^k-1)=1$ and $a \not \in \F_{2^k}^{*}$  & \cite{Wu} \\
\hline
$H_{8}$& $X^{6q-5}+aX$ & see Theorem 3.1  & \cite{LFW}\\
\hline
\end{longtable}


\begin{thebibliography}{99}
\bibitem{AGW} A. Akbary, D. Ghioca, Q. Wang, {\it On constructing permutations of finite fields,} Finite Fields Appl. 17 (2011) 51--67.

\bibitem{AW} A. Akbary, Q. Wang, {\it On Polynomials of the Form $X^rf(X^{\frac{q-1}{l}})$}, Internat. J. Math. Math. Sci. (2007).

\bibitem{bai} T. Bai, Y. Xia, {\it A new class of permutation trinomials constructed from niho exponents,}Cryptogr. Commun. 10, 1023--1036 (2018).

\bibitem{BZ} L.A. Bassalygo, V.A. Zinoviev, {\it Permutation and complete permutation polynomials,} Finite Fields Appl.33, 198--211 (2015).

\bibitem{Bertoni_Daemen} G. Bertoni, J. Daemen, M. Peeters, G. Van Assche, {\it Permutation-based encryption, authentication and authenticated encryption,} DIAC 2012, July 2012.

\bibitem{BS} S. Bhattacharya, S. Sarkar, {\it On some permutation binomials and trinomials over $\F_{2^n}$,} Des. Codes Cryptogr.82, 149--160 (2017).

\bibitem{Bracken_14} C. Bracken, C. H. Tan, Y. Tan, {\it On a class of quadratic polynomials with no zeros and its application to APN functions,} Finite Fields Appl. 25 (2014) 26--36.


\bibitem{CW} L. Carlitz, C. Wells, {\it The number of solutions of a special system of equations in a finite field,} Acta Arith. 12, 77--84 (1966–1967).

\bibitem{Chen} C. Chen, H. Kan, J. Pang, L. Zheng, Y. Li, {\it Three classes of permutation quadrinomials in odd characteristic,} Cryptogr. Commun. 16, 351--365 (2024).

\bibitem{Deng} H. Deng, D. Zheng, {\it More classes of permutation trinomials with Niho exponents,} Cryptogr. Commun. 11, 227--236 (2019).

\bibitem{Ding_C_13}  C. Ding, T. Helleseth, {\it Optimal ternary cyclic codes from monomials,} IEEE Trans. Inf. Theory 59 (2013) 5898--5904.

\bibitem{Ding_Co_06}  C. Ding, J. Yuan, {\it A family of skew Hadamard difference sets,} J. Comb. Theory, Ser. A 113 (2006) 1526--1535.

\bibitem{Ding2020} Z. Ding, M. E. Zieve, {\it Low-degree permutation rational functions over finite fields,} Acta Arithmetica 202 (2022) 253--280.

\bibitem{Dobbertin_niho}  H. Dobbertin, {\it Almost perfect nonlinear power functions on $GF(2^n)$: the Niho case,} Inf. Comput. 151 (1999) 57--72.

\bibitem{Dwokin_SHA3} M. J. Dworkin.   {\it SHA-3 Standard:  Permutation-Based Hash and Extendable-Output Functions,} Federal Information Processing Standards Publication, NIST FIPS - 202, August 2015.

\bibitem{Ferraguti2020} A. Ferraguti, G. Micheli, {\it Full Classification of permutation rational functions and complete rational functions of degree three over finite fields,} Des. Codes Cryptogr, 88(5) (2020) 867--886.

\bibitem{Gupta20} R. Gupta, {\it Several new permutation quadrinomials over finite fields of odd  characteristic.} Des. Codes Cryptogr. 88(1), (2020) 223--239.

\bibitem{Gupta22} R. Gupta, {\it More results about a class of quadrinomials over finite fields of odd characteristic,} Commun. Algebra 50(1) (2022) 324--333.

\bibitem{SG} R. Gupta, R.K. Sharma, {\it Determination of a type of permutation binomials and trinomials,} Appl. Algebra Eng. Commun. Comput. 31, (2020) 65--86.

\bibitem{HXX} X. Hou, {\it A class of permutation binomials over finite fields}. J. Number Theory 133, 3549--3558 (2013).

\bibitem{Hou_PP_15} X. Hou, {\it Permutation polynomials over finite fields-a survey of recent advances,} Finite Fields Appl. 32 (2015) 82--119.

\bibitem{Hou_bi_tri}  X. Hou, {\it A survey of permutation binomials and trinomials over finite fields,} in: G. Kyureghyan, G. L. Mullen, A. Pott (Eds.), Topics in Finite Fields, Proceedings of the 11th International Conference on Finite Fields and Their Applications, vol. 632, AMS, 2015, pp. 177--191.

\bibitem{HouD4} X. Hou, {\it Rational functions of degree four that permute the projective line over a finite}, Commun. Algebra 49(9) (2021) 3798--3809.

\bibitem{Hou_carlitz} X. Hou, {\it A power sum formula by Carlitz and its applications to permutation rational functions of finite fields,} Cryptogr. Commun. 13 (2021) 681--694.

\bibitem{HL} X. Hou, S.D. Lappano, {\it Determination of a type of permutation binomials over finite fields,} J. Number Theory 147, 14--23 (2015).

\bibitem{HV} X. Hou, V.P. Lavorante, {\it New results on permutation binomials of finite fields}, Finite Fields Appl. 88 (2023): 102179.

\bibitem{JLTZ} Y. Jiang, Y. Li, Z. Tu, X. Zeng, {\it Binomial permutations over finite fields with even characteristic,} Des. Codes Cryptogr. 89(12), (2021), 2869--2888.

\bibitem{KX} F. Kousar, M. Xiong, {\it Some permutation pentanomials over finite fields of even characteristic}, (2024) \url{https://doi.org/10.48550/arXiv.2412.14641}.

\bibitem{Chapuy_C_07} Y. Laigle-Chapuy, {\it Permutation polynomials and applications to coding theory,} Finite Fields Appl. 13 (2007) 58--70.

\bibitem{LS} S.D. Lappano, {\it A note regarding permutation binomials over $\F_{q^2}$,} Finite Fields Appl. 34, 153--160 (2015).

\bibitem{LFW} Y. Li, X. Feng, Q. Wang, {\it Towards a classification of permutation binomials of the form $x^i+ ax$ over $\F_{2^n}$} Des. Codes Cryptogr. (2024): 1--17.

\bibitem{KLi_quad_2021} K. Li, C. Li, T. Helleseth, L. Qu. {\it Cryptographically strong permutations from the butterfly structure,} Des. Codes Cryptogr. 89, 737–-761 (2021).

\bibitem{LQ} K. Li, L. Qu, X. Chen, {\it New classes of permutation binomials and permutation trinomials over finite fields,} Finite Fields Appl. 43, 69--85 (2017).

\bibitem{KLi_quad_2020} K. Li, L. Qu, C. Li, H. Chen, {\it On a conjecture about a class of permutation quadrinomials,} Finite Fields Appl. 66 (2020) 101690.

\bibitem{LQW} K. Li, L. Qu, Q. Wang, {\it New constructions of permutation polynomials of the form $x^{r}h\left(x^{q-1}\right) $ over $\F_{q^2}$,} Des. Codes Cryptogr. 86, 2379--2405 (2018).

\bibitem{NLi_quad_2021}  N. Li, M. Xiong, X. Zeng, {\it On permutation quadrinomials and $4$-uniform BCT,} IEEE Trans. Inf. Theory 67 (2021) 4845--4855.

\bibitem{Lidl_Cr_84}  R. Lidl, W.B. Mullen, {\it Permutation polynomials in RSA-cryptosystems,} in: Advances in Cryptology, Plenum, New York, 1984, pp. 293--301.

\bibitem{WL}  R. Lidl, D. Wan, {\it Permutation polynomials of the form $X^r f (X^\frac{q-1}{d})$ and their group structure,} Monatshefte Math. 112, 149--163 (1991).

\bibitem{QLiu_penta_2023} Q. Liu, G. Chen, X. Liu, J. Zou, {\it Several classes of permutation pentanomials with the form $X^rh(X^{p^m- 1})$ over $\F_{p^{2m}}$,} Finite Fields Appl. 92 (2023) 102307.

\bibitem{Gupta_resi_2023} F.E.B. Mart{\'\i}nez, R. Gupta, L. Quoos, {\it Classification of some permutation quadrinomials from self reciprocal polynomials over $\F_{2^n}$,} Finite Fields Appl. 91 (2023) 102276.

\bibitem{MZ} A.M. Masuda, M. E. Zieve, {\it Permutation binomials over finite fields,}. Trans. Am. Math. Soc. 361, 4169--4180 (2009).

\bibitem{OM} J.A. Oliveira , F.E.B. Mart{\'\i}nez, {\it Permutation binomials over finite fields,} Discrete Math. 345(3), (2022), 112732.


\bibitem{FB_quad_2023} F. {\"O}zbudak, B. G{\"u}lmez Tem{\"u}r, {\it Classification of some quadrinomials over finite fields of odd characteristic,} Finite Fields Appl. 87 (2023) 102158.

\bibitem{Park_2001} Y. H. Park, J. B. Lee, {\it Permutation polynomials and group permutation polynomials,} Bull. Austral. Math. Soc. 63(1) (2001), 67--74.

\bibitem{Rai} A. Rai, R. Gupta, {\it Permutation polynomials of the form $x^rh(x^{q- 1})$ over $\F_{q^2}$ with even characteristics} Finite Fields Appl. 104 (2025): 102594.

\bibitem{Schwenk_Cr_98} J. Schwenk, K. Huber, {\it Public key encryption and digital signatures based on permutation polynomials,} Electron. Lett. 34 (1998) 759--760.

\bibitem{Shen} R. Shen, X. Liu, X. Xu, {\it More constructions of permutation pentanomials and hexanomials over $\F_{p^{2m}}$}, Appl. Algebra Engrg. Comm. Comput., \url{https://doi.org/10.1007/s00200-024-00673-3} (2024).
\bibitem{CSze} {C. Sze, {\it Rational functions of degree five that permute the projective line over a finite field}, Ph.D.Thesis, University of South Florida, 2023.}

\bibitem{ZTu_quad_19} Z. Tu, X. Liu, X. Zeng, {\it A revisit to a class of permutation quadrinomials,} Finite Fields Appl. 59 (2019) 57--85.

\bibitem{ZTu_quad_18_2} Z. Tu, X. Zeng, T. Helleseth, {\it New permutation quadrinomials over $\F_{2^{2m}}$}, Finite Fields Appl. 50 (2018) 304--318.

\bibitem{ZTu_quad_18_1}  Z. Tu, X. Zeng, T. Helleseth, {\it A class of permutation quadrinomials,} Discrete Math. 341 (2018) 3010--3020.



\bibitem{TG} G. Turnwald, {\it Permutation polynomials of binomial type,} Contributions to General Algebra, vol. 6, 281--286. HRolder-Pichler-Tempsky, Vienna (1988).

\bibitem{Wang_2007} Q. Wang, {\it Cyclotomic mapping permutation polynomials over finite fields.} In: S. W. Golomb, G. Gong, T. Helleseth, H. Y. Song, eds. Sequences, Subsequences, and Consequences, in: Lect. Notes Comput. Sci. Berlin: Springer, vol. 4893, pp. 119--128, 2007.

\bibitem {WangIndex} Q. Wang, {\it Polynomials over finite fields: an index approach}, in the Proceedings of Pseudo-Randomness and Finite Fields, Multivariate Algorithms and their Foundations in Number Theory, October 15-19, Linz, 2018, Combinatorics and Finite Fields. Difference Sets, Polynomials, Pseudorandomness and Applications, Degruyter, 2019, pp. 319-348.

\bibitem{Wil} K.S. Williams, {\it Note on cubics over $\F_{2^n}$and $\F_{3^n}$,} J. Number Theory 7(4) (1975), 361--365.

\bibitem{Wu} G. Wu, N. Li. T. Helleseth, Y. Zhang, {\it Some classes of monomial complete permutation polynomials over finite fields of characteristic two}, Finite Fields Appl. 28 (2014): 148--165.

\bibitem{wu2017} D. Wu, P. Yuan, C. Ding, Y. Ma,   \textit{Permutation trinomials over $\mathbb{F}_{2^m}$}, Finite Fields Appl. 46 (2017) 38--56.

\bibitem{Xu_penta_18} G. Xu, X. Cao, J. Ping, {\it Some permutation pentanomials over finite fields with even characteristic,} Finite Fields Appl. 49 (2018) 212--226.

\bibitem{Zhang_penta_23} T. Zhang, L. Zheng, X. Hao, {\it More classes of permutation hexanomials and pentanomials over finite fields with even characteristic,} Finite Fields Appl. 91 (2023) 102250.

\bibitem{Zheng_quad_22} L. Zheng, B. Liu, H. Kan, J. Peng, D. Tang, {\it More classes of permutation quadrinomials from Niho exponents in characteristic two,} Finite Fields Appl. 78 (2022) 101962.

\bibitem{Zieve2008} M. E. Zieve, {\it Some families of permutation polynomials over finite fields.} Int. J. Number Theory 4 (2008), 851--857.

\bibitem{Zieve2013} M. E. Zieve, {\it Permutation polynomials on $\F_q$ induced from R\'edei function bijections on subgroups of $\F_q^*$}, Monatsh. Math., to appear, arXiv (2013). {\url{https://arxiv.org/abs/1310.0776}}
 \end{thebibliography}
\end{document}